\documentclass[superscriptaddress,twocolumn, nofootinbib]{revtex4-1}
\usepackage{lipsum}
\usepackage[utf8]{inputenc} 
\usepackage[hidelinks]{hyperref}    
\usepackage{url}            
\usepackage{booktabs}       
\usepackage{nicefrac}      
\usepackage{lipsum}		
\usepackage{graphicx}
\usepackage{mathtools}
\usepackage{xcolor}
\usepackage{amsmath}
\usepackage{amssymb, amsthm}
\usepackage{amsfonts}       
\usepackage{physics}
\usepackage{xspace}
\usepackage{bigints}
\usepackage{appendix}
\usepackage{mathtools}
\usepackage{subfig}
\usepackage{paralist}
\usepackage{tikz,tikz-3dplot}
\usepackage{quantikz}
\usepackage[linesnumbered, ruled,vlined]{algorithm2e}
\usepackage{bm}
\usepackage{framed}
\usepackage{cleveref}
\usepackage[shortlabels]{enumitem}
\usepackage{subfig}
\usepackage{textpos}

\newcommand{\mbrc}{\ensuremath{{\rm MBR}_C}}
\newcommand{\mbrclog}{\ensuremath{{\rm MBR}_{C-\log}}}
\newcommand{\mbrq}{\ensuremath{{\rm MBR}_{\not{C}}}}
\newcommand{\mbrmub}{\ensuremath{{\rm MBR}_{\rm MUB}}}
\newcommand{\mps}{\ensuremath{{\rm MPS}}~}
\newcommand{\stab}{\ensuremath{{\rm STAB}}~}
\newcommand{\sparse}{\ensuremath{{\rm SPARSE}}~}

\newtheorem{theorem}{Theorem}
\newtheorem{definition}{Definition}
\newtheorem{corollary}{Corollary}

\newtheorem{lemma}{Lemma}

\newcommand{\bigO}{\ensuremath{\mathcal{O}}\xspace}

\renewcommand{\Tr}[1]{\ensuremath{\operatorname{Tr}\left({#1}\right)}}

\newcommand{\cbe}{\ensuremath{CB_\varepsilon}}

\DeclarePairedDelimiterX{\inp}[2]{\langle}{\rangle}{#1, #2}

\renewcommand{\epsilon}{\varepsilon}
\newcommand{\epsrank}{$\epsilon$-rank }

\newcommand{\ryc}{reduce\&chop\ }
\newcommand{\Opoly}[1]{\ensuremath{\mathcal{O}\left({\rm poly}\left(#1\right)\right)}}
\newcommand{\maxF}{\ensuremath{\max_{a\neq b} \Vert F_{a, b} \Vert_\infty}}

\newcommand{\aqa}{$\langle aQa ^L\rangle $ Applied Quantum Algorithms, Universiteit Leiden}
\newcommand{\lorentz}{Instituut-Lorentz, Universiteit Leiden, Niels Bohrweg 2, 2333 CA Leiden, Netherlands}
\newcommand{\liacs}{LIACS, Universiteit Leiden, Niels Bohrweg 1, 2333 CA Leiden, Netherlands}

\newcommand{\para}[2]{\textbf{#1}: #2}
\renewcommand{\para}[2]{#2}

\begin{document}
\title{Multiple-basis representation of quantum states}
\author{Adrián Pérez-Salinas} 
\author{Patrick Emonts}
\author{Jordi Tura} 
\affiliation{\aqa}
\affiliation{\lorentz}
\author{Vedran Dunjko}
\affiliation{\aqa}
\affiliation{\liacs}

\begin{abstract}
Classical simulation of quantum physics is a central approach to investigating physical phenomena. 
Quantum computers enhance computational capabilities beyond those of classical resources, but it remains unclear to what extent existing limited quantum computers can contribute to this enhancement. 
In this work, we explore a new hybrid, efficient quantum-classical representation of quantum states, the multiple-basis representation.
This representation consists of a linear combination of states that are sparse in some given and different bases, specified by quantum circuits. 
Such representation is particularly appealing when considering depth-limited quantum circuits within reach of current hardware.
We analyze the expressivity of multiple-basis representation states depending on the classical simulability of their quantum circuits.
In particular, we show that multiple-basis representation states include, but are not restricted to, both matrix-product states and stabilizer states. 
Furthermore, we find cases in which this representation can be used, namely approximation of ground states, simulation of deeper computations by specifying bases with shallow circuits, and a tomographical protocol to describe states as multiple-basis representations.
We envision this work to open the path of simultaneous use of several hardware-friendly bases, a natural description of hybrid computational methods accessible for near-term hardware.
\end{abstract}
\maketitle

\section{Introduction}

Classical simulation of quantum processes is broadly utilized approach to investigate physics, such as material discovery~\cite{ma2020quantum}, electronic structure problems~\cite{mcclean2019openfermion} or chemistry~\cite{genin2019quantum}. 
Through the years, a large variety of methods has been developed to efficiently simulate certain classes of quantum states, such as tensor networks~\cite{cirac2021matrix,schollwock2011densitymatrix}, the stabilizer formalism~\cite{gottesman1998heisenberg, aaronson2004improved, bravyi2019simulation}, or decision diagrams~\cite{vinkhuijzen2023limdd} among others~\cite{goh2023liealgebraic, cerezo2023does,rudolph2023classical}. ulation methods pushes the borders of current knowledge to represent quantum states using classical resources, but there still exist quantum states of interest that elude an efficient characterization by existing classical methods. 

Quantum computers offer new computational capabilities as compared to classical machines. 
As a motivation, quantum computing allows, in principle, the simulation quantum processes beyond the possibilities of classical computers~\cite{feynman1982simulating}. 
In addition, there exist certain problems, outside the domain of physics,where it is expected that the best classical algorithms will always require super-polynomial runtime in the size of the problem, and that have been proven to become efficiently solvable with quantum resources \cite{shor1997polynomialtime,liu2021rigorous,molteni2024exponential,harrow2009quantum,huang2021quantum}.
Nowadays existing prototypes of quantum devices are pushing the frontier of addressable quantum computations, in the quest for quantum advantage~\cite{arute2019quantum,madsen2022quantum,zhong2020quantum}. 
However, these experimental devices still suffer from limitations in qubit numbers, depth or fidelity of operations. 
Recent research has explored the combination of limited quantum resources with classical algorithms to harness the full potential of existing quantum devices~\cite{piveteau2023circuit, huembeli2022entanglement, peng2020simulating, marshall2022high, perez-salinas2023shallow}.

In this work, we explore a new representation of quantum states built on the idea of leveraging limited quantum resources for reaching certain classically inaccessible computations. 
We call our construction multiple-basis representation (MBR). 
The ingredients for the MBR are a $\Opoly n$-sized set ($n$ being the number of qubits) of classical parameters, and a set of $\Opoly n$-sized quantum circuits. 
Each such quantum circuit determines a basis.
The MBR is built as a superposition of states from each of these bases, each state admitting sparse descriptions, i.e. with $\Opoly{n}$ many coefficients, in the corresponding bases.  
Given these coefficients and the quantum circuits specifying the bases, access to a quantum computer is required to compute relevant quantities of the state of interest, for instance energies, using the MBR. 
The MBR then relies on a classical combination of quantum resources.
The MBR enables a systematic recipe to obtain approximations of states of interest, such as ground states. 
Additionally, it is possible to relax the conditions on the quantum resources of MBR until making the computation accessible to classical computer, at the expense of losing generality. 
Therefore, as a byproduct of our construction, we obtain a novel fully classical method to represent quantum states.

The goal of MBR is to exhibit strictly more representation capabilities than single-basis approaches. 
These representation capabilities will be influenced by the choice of bases in the MBR. 
We explore the conditions these bases must fulfill to deliver MBRs with the following properties. 
(i) As a first requirement, the bases
should be chosen such that each of them delivers some information of the global state that is inaccessible to all other
considered bases (maintaining the sparsity constraint). (ii) Additionally, the MBR should be made inaccessible to classical simulation methods. The MBR carries a natural notion of optimality with respect to condition (ii), which we show it is met when the chosen basis are MUBs.

We show how MBR can be used in several tasks of interest. 
We explore first the approximation of ground states of physically motivated Hamiltonians.
We assume that the Hamiltonian admits the form $H=\sum_i h_i$ such that each term $h_i$ can be easily diagonalized, and $[h_i, h_j] \neq 0$ (possibly with different bases).
The results from these individual terms $h_i$ can be combined through the MBR formalism, yielding approximations to the ground state. 
Our numerical study suggests the MBR method could be used to approximate ground states of many-body physics problems on arbitrary graphs with an advantage over tensor networks approaches.
The second application is the simulation of deeper circuits with shallow computations, as a follow up of Ref.~\cite{perez-salinas2023shallow}. 
The MBR formalism allows for a further reduction in depth of the circuits of interest. 
Finally, we propose a tomography protocol of quantum states using the MBR form.
This tomography protocol is related to a relaxed notion of peaked circuits. Peaked circuits are interesting, since they admit efficient verification of quantum advantages~\cite{aaronson2024verifiable}.
These three applications enhance the computational capabilities of hybrid quantum classical computations, following the recently open research path on classical combination of limited quantum computations~\cite{piveteau2023circuit,huembeli2022entanglement,marshall2022high}.

This paper is organized as follows. 
\Cref{sec.preliminaries} contains the necessary background for the rest of the manuscript.  
In~\Cref{sec.mbr}, we present the MBR and define several subclasses of MBR, exploring their relationships among them and with other classical representations of quantum states, and optimal choices of bases. 
Three distinct applications, ground state approximation, optimization of quantum algorithms, and tomography, are described in~\Cref{sec.applications}. 
Finally, we conclude in~\Cref{sec.conclusions}.

\section{Preliminaries}\label{sec.preliminaries}

\subsection{Classical representations of quantum states}\label{sec.representations}

We begin our discussion by looking at numerous classical representations of pure quantum states from the perspective of various \textit{ranks} that will become clear from context. 
The rank of a quantum state or unitary, with respect to certain decomposition criteria, is the minimal number of components needed to exactly represent the state of interest in that particular decomposition. 
Classical representability of quantum states crucially depends on the concept of rank. 
By convention, for a quantum state $\ket\psi$ to admit an efficient classical representation, there necessarily exist a decomposition in which a) $\ket\psi$ has a \Opoly n rank, and b) each component in which $\ket\psi$ is decomposed admits an efficient classical representation.

As an introduction to the arguments in the main paper we describe in this section three classical representations of quantum states: sparse states, matrix product states (MPS) and stabilizer representations of quantum states.
These methods are well-known simulation methods for quantum computing, but the list is far from exhaustive. 
Approaches such as $\mathfrak g$-sim~\cite{goh2023liealgebraic}, LOWESA~\cite{rudolph2022decomposition}, specifically designed for variational algorithms, or combinations between stabilizer and tensor networks methods~\cite{masot-llima2024stabilizer} are also available in literature.

\subsubsection{Sparse representation}

The first rank we discuss is the computational-basis (CB) rank of a quantum state, measuring how sparse a state is with respect to the computational basis.  
\begin{definition}[$\cbe$-rank~\cite{perez-salinas2023shallow}]\label{def.cbe}
Let $\ket{\psi}$ be an $n$-qubit quantum state, and let $\epsilon > 0$. 
The $\epsilon$-approximate $CB$ rank of a state $\ket\psi$ is 
\begin{equation}
\cbe(\ket\psi) := \min_{K\in\mathbb{N}}\left(K \; {\rm s. t.}\;  \max_{\ket{\psi_K} \in \Phi_K}\vert\braket{\psi}{\phi_{K}}\vert^2 \geq 1 - \varepsilon\right),
\end{equation}
with 
\begin{equation}
    \Phi_K = \left\{\ket{\phi_K} \, \left\vert \, \ket{\phi_K} = \sum_{\substack{i \in S \\ \vert S \vert = K}} \alpha_i \ket i\right. \right\}, \end{equation}
     and $S \subset \{0, 1, 2, \ldots, 2^n - 1 \}$.
\end{definition}
It can be further shown that the optimal $\ket{\phi_K}$ in~\Cref{def.cbe} is obtained by retaining the $K$ elements in $\ket\psi$ with largest amplitudes and discarding everything else, with proper normalization~\cite{perez-salinas2023shallow}. 
Note that it is possible to extend the definition of the \cbe-rank to the rank in any other basis. 
Consider the basis $\{U\ket i\}_i$, with $U$ unitary and $\ket i$ being the elements of the computational basis. 
Then, the \cbe-rank in the basis specified by $U$ is just $\cbe(U^\dagger \ket\psi)$. 
In the remainder of the paper $CB(\ket\psi) \equiv CB_{\epsilon = 0}(\ket\psi)$. This extends to other ranks. States of the form $\ket{\phi_K}$ are referred to as $K$-sparse states.

Small \cbe-ranks enable classical simulations of quantum computations by storing only sparse states. 
Consider a quantum circuit $U = U_{L} U_{L - 1} \ldots U_2 U_1$, where $U_l$ are arbitrary $\mathcal O(1)$-local gates,  such that for all $l \leq L$, $\cbe\left(U_{l} \ldots U_2 U_1 \ket 0 \right) \in \Opoly n$. 
Under this assumption, the simulation can be conducted by updating the sparse representation of the state at each layer. 
We refer to this concept as sparse simulation. 
\begin{definition}[\sparse simulation]
    Let $U = U_L U_{L - 1} \ldots U_2 U_1$ be a quantum circuit of $n$ qubits, such that $U_l$ are arbitrary operations acting on $\mathcal O(1)$ qubits. 
    $U$ is \sparse-simulatable with precision $\epsilon$ and with respect to an observable $\hat O$ if 
    \begin{equation}
        \cbe(U_{l} \ldots U_2 U_1 \ket 0) \in \Opoly n, \quad \forall l \leq L, 
    \end{equation}
    and the \sparse representation of $U\ket 0$ permits to efficiently compute expectation values of this state with respect to the  observable $\hat O$. 
\end{definition}

\subsubsection{Matrix product state representation}

The Schmidt rank is our second rank to review, closely related to matrix product states. 
The Schmidt rank provides a natural decomposition for studying correlations between quantum subsystems, as well as a robust framework for understanding bipartite entanglement. 
To define the Schmidt rank, we need to first partition the available Hilbert space into two subspaces, as $\mathcal H = \mathcal H_A\otimes \mathcal H_B$, and then perform the Schmidt decomposition.  
\begin{definition}[Schmidt rank]
    Let $\ket{\psi}$ be an $n$-qubit quantum state, and let $\epsilon>0$. 
    Let the Hilbert space be decomposed into two partitions $(A, \bar A)$, as $\mathcal H = \mathcal H_A \otimes \mathcal H_{\bar A}$. 
    The $\epsilon$-approximate Schmidt rank of a state $\ket\psi$ is
    \begin{equation}
        \chi^{(A)}_\epsilon(\ket\psi) = \min_{K \in \mathbb N}\left(K \; {\rm s. t.}\;  \max_{\ket{\phi_K} \in \Phi_K}\vert\braket{\psi}{\phi_{K}}\vert^2 \geq 1 - \varepsilon\right),
    \end{equation}
    with
    \begin{equation}
        \Phi_K = \left\{ \ket{\phi_K} \left \vert \ket{\phi_K} = \sum_{k = 1}^K \lambda_k \ket{u_k}_A \otimes \ket{v_k}_{\bar A} \right. \right\}, 
    \end{equation}
    and $\lambda_k \geq 0, \braket{v_k}{v_l} = \delta_{k, l}, \braket{u_k}{u_l} = \delta_{k, l}$.
\end{definition}
Notice that the Schmidt rank may vary depending on the partition $A$. 
It is upper-bounded by $\chi_{(A)} = 2^{\min(\vert A \vert, \vert \bar A \vert)}$. 

The Schmidt decomposition has a direct application in simulation of quantum states and circuits via tensor networks~\cite{schollwock2011densitymatrix}, and in particular matrix product states (MPS)~\cite{cirac2021matrix}. 
An MPS is a representation of quantum states in which the qubits are ordered in a one dimensional chain. 
The partitions $A, B$ are defined by establishing a border between \textit{left} and \textit{right}. 
An MPS is constructed by chaining Schmidt decompositions in this order, as 
\begin{multline}
    \ket{\psi_{\rm MPS}} := \sum_{i_1 \ldots i_n} \sum_{j_1 \ldots j_{n}} \\ 
    A^{(1), i_1}_{j_1} A^{(2), i_2}_{j_1, j_2} \ldots A^{(n-1), i_{n-1}}_{j_{n-1}, j_n} A^{(n),  i_n}_{j_n} \ket{i_1 i_{2} \ldots i_{n-1} i_n},
\end{multline}
with $A^{(k), i_k}_{j_{k-1} j_k}$ being a tensor of size $d \times D \times D$, with $d$ being the local dimension and $D$ being the so-called bond dimension of the MPS. 
The indices $i \in \{ 1,2, \ldots d\}$ are related to output states, usually in the computational basis. 
The indices $j$ are associated to the correlations among partitions. 
For the MPS to be an efficient classical representation of a quantum states $D \in \Opoly n$. 
For many interesting cases, this constraint still provides an accurate approximation for MPS to constitute an useful technique. 

\begin{definition}[\mps simulation]
    Let $U = U_L U_{L - 1} \ldots U_2 U_1$ be a quantum circuit on $n$ qubits, where $U_l$ are arbitrary gates acting on $\mathcal O(1)$ gates. 
    $U$ is \mps simulatable with precision $\epsilon$ and with respect to an observable $\hat O$ if for any bipartition $(A, \bar A)$ 
    \begin{equation}
        \chi_\epsilon^{(A)}(U_{l} \ldots U_2 U_1 \ket 0) \in \Opoly n, \quad \forall l \leq L, 
    \end{equation}
    and the \mps representation of $U\ket 0$ permits to efficiently compute expectation values of this state with respect to $\hat O$. 
\end{definition}
Computing expectation values through MPS is possible as long as the observable of interest admits a Schmidt decomposition with $\chi_\epsilon^{(A)} \in \Opoly n$, what is usually called a matrix product operator, for instance for Pauli strings. 
The \mps method allows for classical simulations of quantum circuits with nearest-neighbor connectivity and $\mathcal O(\log n)$ depth.

\subsubsection{Stabilizer representation}

Finally, we address the stabilizer rank, naturally connecting quantum states with the stabilizer formalism~\cite{gottesman1998heisenberg}. 
The stabilizer rank captures the optimal decomposition of state $\ket\psi$ into a superposition of states that are outputs of quantum circuits composed only of Clifford gates, that is the group of unitary operations that stabilizes the Pauli group. 
\begin{definition}[Stabilizer rank]
    Let $\ket{\psi}$ be an $n$-qubit quantum state, and let $\epsilon>0$. 
    The $\epsilon$-approximate stabilizer rank of a state $\ket\psi$ is
    \begin{equation}
        \sigma_\epsilon(\ket\psi) := \min_{K \in \mathbb N}\left(K \; {\rm s. t.}\;  \max_{\ket{\phi_K}}\vert\braket{\psi}{\phi_{K}}\vert^2 \geq 1 - \varepsilon\right),
    \end{equation}
    with
    \begin{equation}
        \ket{\phi_K} = \sum_{k = 1}^K \beta_k U_k \ket 0, 
    \end{equation}
    and $U_k$ being circuits composed only by Clifford gates. 
\end{definition}
The stabilizer rank is hard to compute~\cite{bravyi2019simulation}, but it can be easily bounded with the number of non-Clifford gates needed to prepare $\ket\psi$.  
The reason is that Clifford circuits allow one to perform quantum computations by keeping track of the updates in Pauli string stabilizing the initial state~\cite{gottesman1998heisenberg}. 
These operations are classically accessible. When a non-Clifford gate appears, e.g. a $T$-gate, one can extend the stabilizer formalism to decompose the gate as a superposition of Clifford operations, thus effectively increasing $\sigma_\epsilon(\ket\psi)$ by a constant factor. Hence, circuits with $\mathcal O(\log n)$ non-Clifford gates admit classical simulations~\cite{gottesman1998heisenberg}.

\begin{definition}[\stab simulation]
    Let $U = U_L U_{L - 1} \ldots U_2 U_1$ be a quantum circuit of $n$ qubits, where $U_l$ act on at most $\mathcal O(1)$ qubits. 
    $U$ is \stab simulatable with precision $\epsilon$ and with respect to an observable $\hat O$ if $U$ has at most $\mathcal O(\log T)$ non-Clifford gates, and $\hat O$ is decomposable on at most $\Opoly n$ Pauli strings. Under that assumption, 
    \begin{equation}
        \sigma_\epsilon(U_{l} \ldots U_2 U_1 \ket 0) \in \Opoly n, \quad \forall l \leq L. 
    \end{equation}
\end{definition}

\subsection{Mutually unbiased bases}\label{sec.mubs}

Throughout this paper the concept of mutually unbiased bases (MUB) will become useful to specify the a subser of MBR with optimal properties in some sense, to be detailed. 
MUBs are collections of bases such that any element of any basis can only be written as an equal superposition of all elements of other basis. 
The formal definition is as follows. 
\begin{definition}[Mutually unbiased bases~\cite{durt2010mutually}]
    Let $\{ U_b\}$ be a set of $N \times N$ unitary matrices, each one defining the basis $\{U_b\ket i\}$, for $\ket i$ in the computational basis. 
    The set $\{U_b\ket i\}$ is a set of mutually unbiased basis if
    \begin{equation}
        \forall (i, j); \forall (a, b), a \neq b, \qquad \vert \bra j U_a^\dagger U_b \ket i\vert^2 = \frac{1}{N}. 
    \end{equation}
\end{definition}
Sets of MUBs in a $N$-dimensional system have cardinality bounded by $2 \leq \vert \{ U_b\}\vert\leq N + 1$~\cite{klappenecker2003constructions}. 
If the upper bound can be saturated, then the MUBs form a complete set. 
Existence of complete sets of MUBs are guaranteed if $N$ is a prime power, but remains an open problem for general $N$~\cite{grassl2009sicpovms}. 
There exist numerous ways to find MUBs~\cite{bandyopadhyay2001new, kern2010complete, klappenecker2003constructions, seyfarth2011construction, seyfarth2015practical}. 
It is in fact possible to construct complete sets of MUBs using only Clifford circuits \cite{kern2010complete}, thus thus further motivating the study of the \stab\ representation of quantum states. 
Let us remark that there exist non-Clifford constructions for complete sets of MUBs, for instance through the bases defined as $V U_b$, where $\{U_b\}$ are MUBs and $V$ is an arbitrary non-Clifford unitary.

\section{Multiple-Bases Representation of quantum states}\label{sec.mbr}
The multiple-basis representation (MBR) of a quantum state is the central element of this work. 
The MBR aims to approximately describe an arbitrary quantum state with a few elements combining classical resources and quantum operations.

Inspired by a combination of the \cbe- and \stab-ranks, we define the MBR as follows. 
\begin{definition}[Multiple-basis representation (MBR)]\label{def.mbr}
Let $\{U_b\}_{b = 1}^B$ be a set of quantum circuits, $\{\alpha_b\}$ be a set of real positive values, $\{ c^{(b)}_{i_b}\}$ a set of complex coefficients and $\{S_b\}$ a set of supports of cardinality $K$, independent for each basis. The values $K, B \in \Opoly n$.
A MBR state is given by
\begin{align}\label{eq.state_def}
\ket{\psi_{K, B}} & = \sum_{b= 1}^B \alpha_b U_b \ket{\psi_b} \\
{\rm with } \quad \ket{\psi_b} & = \sum_{\substack{i_b \in S_b \\ \vert S_b \vert = K}} c^{(b)}_{i_b} \ket{i_b} \\ 
{\rm and} \quad \braket{\psi_b}{\psi_b}  & = 1 \label{eq.normalization},
\end{align}
with $\ket{i_b}$ being elements of the computational basis. 
\end{definition} 
The interpretation of a MBR states can be done as follows. We start with the states $\ket{\psi_b}$. Those will be states with low $\cbe$-ranks, and thus have limited power to represent all quantum states. The role of $U_b$ is to perform a unitary operation that changes the basis in which the state is sparse. Each new term in \Cref{eq.state_def}, noted with the index $b$, will add a new state which is sparse in a completely new basis $U_b$, but not sparse in any of the previously considered bases. The superposition of sparse states in different bases is intended to create a new state which is not sparse in any of the previous bases, and hence does not admit an efficient representation. On the other hand, there must exist bases such that $\ket{\psi_{K, B}}$ is sparse in these bases, even with unit \cbe-rank. However, we expect that constructing these bases will require much deeper circuits than each of the individual unitaries $U_b$. A graphical interpretation of MBR is available in~\Cref{fig.mbr}. 

We now detail the requirements to fully specify a MBR. 
The unitary gates $\{ U_b\}$ must be given in an efficient form, for instance through a $\Opoly n$-sized circuit. 
For each basis we need $\mathcal O(K)$ memory to store (up to digital precision) $\alpha_b, \{ c^{(b)}_{i_b}\}$, and the set $S_b$. 
By convention we will order $S_b$ so that $\vert c^{(b)}_{i_b}\vert^2 \geq \vert c^{(b)}_{i_b + 1}\vert^2 $. 
Notice that these elements are not sufficient to fully specify the state due to the linear dependence of its different components. An additional necessary component will be the overlap matrix
\begin{equation}\label{eq.fidelity_matrix}
        F_{i_b, j_a} = \bra{i_b} U_b^\dagger U_a \ket{j_a}.
\end{equation}
Whether or not this Gram matrix can be computed classically determines whether the MBR can be treated as a fully classical or hybrid classical-quantum representation, although its storage is always possible, by construction. This matrix has a clear structure of $B \times B$ blocks, each of size $K \times K$\footnote{The sizes of these blocks may change in the case the support of each basis has a different number of elements.}. 
The blocks in the diagonal are just the identity matrix relating elements within the same basis. 

\begin{figure}
\includegraphics[scale=.65]{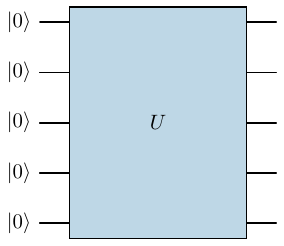} \raisebox{1.2cm}{ \large $ = \sum c^{(b)}_{i_b} $} \includegraphics[scale=.65]{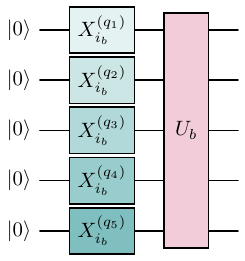} 
\caption{
    A motivation for the MBR is to describe states using superpositions of quantum circuits that are more easily implementable. 
    See~\Cref{def.mbr} for further details.
}\label{fig.mbr}
\end{figure}

The motivation to create the MBR can be easily outlined in terms of the \cbe-rank. 
Take an arbitrary state $\ket\psi$, and a bases $U$. 
By definition, there are certain bases in which $\cbe(U\ket\psi)$ is a large quantity, and some others in which it is small, or even minimal, see Ref.~\cite{perez-salinas2023shallow} for an in-depth discussion. 
There exists a trade-off between allowed \cbe-rank and number of gates needed to implement a proper basis. 
In the MBR representation, we aim to reduce this requirement. 
Each basis is now asked to capture a fraction of the quantum state within a sparse approximation, that is small \cbe-rank for moderately large $\epsilon$. 
The combined action of all bases, each capturing a different fraction of the state, gives representation power to the MBR, allowing us to efficiently describe (certain) states that require exponentially many coefficients using only one basis. 
This change cannot significantly change the volume of $\epsilon$-approximately representable states, as in the \mps\ or \stab\ cases. 
However, it can make new states reachable within modest quantum resources. 
We illustrate this idea in~\Cref{fig.mbr}, a).

\subsection{Example of a MBR}\label{sec.example}
We give an example for a state that offers an efficient MBR representation in this section.
Consider the quantum state given by
\begin{equation}\label{eq.example}
\ket{\psi_{+ 0}} = \frac{1}{\sqrt{2 C}} \left(\ket 0^{\otimes n} + \ket +^{\otimes n} \right),
\end{equation}
where $C = 1 + 2^{-n/2}$ is the normalization constant. 
The \cbe-rank is given by $\cbe(\ket{\psi_{+0}}) = 1 + \max(\varepsilon - C/2, 0)(2^n - 1)$. 
Consequently, an approximate representation of $\ket{\psi_{+0}}$ is achievable even with \cbe-rank 1 only with error above approximately $1/2$. 
More accurate representations demand an exponential \cbe-rank in $n$, since adding any extra component $\ket i, i \neq 0$ can only increase the fidelity by capturing factor $\braket{i}{\psi_{+0}} = 2^{-n/2}$, existing $2^n - 1$ possible values for $i$. Explicitly, 
\begin{multline}
    \cbe(\ket{\psi_{+0}}) = \\
        1 + \left\lfloor \left(1-\epsilon - \frac{1+2^{-n/2}}{2}\right) (2^n - 1)\right\rfloor.
\end{multline}
The Hadamard basis provides no alternative, since $ \ket{\psi_{+0}}$ stabilizes $H^{\otimes n}$. 
Nevertheless, simultaneous use of both the computational and Hadamard basis allows for a sparse representation, as intuited in the explicit writing of~\Cref{eq.example}.

We can explore alternative classical representations for $\ket{\psi_{+0}}$. For instance, it is possible to consider a basis with elements $\ket 0^{\otimes n}$ and $\ket +^{\otimes n} - 2^{-n/2} \ket 0^{\otimes n}$. Preparation of the second state can be done in depth $\mathcal{O}(n)$~\cite{yoder2014fixedpoint} to the required precision $\epsilon \in \mathcal{O}(2^{-n})$.  
Alternatively, one can use matrix product states to exactly represent $\ket{\psi_{+0}}$. 
Such matrix product state has bond dimension 2, and thus it is required to store $\mathcal O(n)$ real numbers. 
Alternatively, one can use a stabilizer representation, with stabilizer rank 2~\cite{aaronson2004improved, bravyi2019simulation}. 
The memory requirement for this construction scales also as $\mathcal O(n)$, achieved by specifying the Clifford circuits of linear size needed to prepare $\ket{\psi_{+0}}$. 
If we address the depth required to prepare this state from the $\ket 0^{\otimes n}$ state, we find $\Omega(\log n)$ through optimal recipes for matrix product states~\cite{malz2023preparation}. 
The MBR in this case is a special case of the stabilizer formalism. 
The minimal depth required to construct a state with small stabilizer rank is, to the best of our knowledge, an open problem.

\subsection{Comparison to other representations}\label{sec.comparison}

In this section, we will connect the MBR to other simulation methods of quantum computing, reviewed in~\Cref{sec.representations}. 
To do so, we will define several subclasses for the MBR, distinguishing between different properties of the circuits, and will compare them to existing methods. 
The first element of discussion will be the classical tractability of a MBR. 
In all this discussion, we restrict ourselves to observables that are Pauli strings or sums of moderately many Pauli strings. 

\begin{definition}[Classical tractability]
Consider a quantum state $\ket\psi$, and an observable $\hat O$. Then, $\ket\psi$ is classically tractable with respect to $\hat O$ if computing the expectation $\bra\psi \hat O \ket\psi$ is efficient with classical resources.
\end{definition}

In the context of MBR we can give a more restricted definition. 
\begin{definition}[Classical tractability of MBR]
    Let $\{U_b, S_b, \bm\alpha_b\}$ be respectively the unitaries, integers sets and coefficients defining a MBR, as in~\Cref{def.mbr}. 
    The MBR is classically tractable, for an observable $\hat O$, if
    it is possible to classically and efficiently compute
    \begin{equation}
        F^{\hat O}_{i_b, j_a} = \bra{i_b} U_b^\dagger \hat O U_a \ket{j_a}.
    \end{equation}
\end{definition}

Note that classical tractability of MBR immediately implies that computing expectations values $\bra{\psi_{K, B}} \hat O \ket{\psi_{K, B}}$ becomes classically feasible since it is just a combination of $\Opoly n$-many elements of $F^{\hat O}_{i_b, j_a}$. 

\begin{lemma}\label{le.exp-value}
    Classical tractability of the MBR implies classical tractability. 
\end{lemma}
The proof is available in~\Cref{app.exp-value}

The MBR definition allows for several subclasses of MBR, depending on the properties of the circuits $U_b$. 
We define four different classes. 

\begin{definition}[MBR - Classical (\mbrc)]
    We define \mbrc\ to be the set \mbrc\ consists of all classically tractable MBRs.
\end{definition}

\begin{definition}[MBR - Non classical (\mbrq)]
    We define \mbrq\ to be the set of all classically intractable MBRs. For \mbrq, $F^{\hat O}_{i_b, j_a}$ is efficiently computable with quantum resources.  
\end{definition}

\begin{definition}[MBR - classical logarithmic (\mbrclog)]
    We define \mbrclog\ to be a subset of \mbrc\ with the constraint that the depth of $\{U_b\}$ is $\mathcal{O}(\log n)$.
\end{definition}

\begin{definition}[MBR - MUBs (\mbrmub)]
    We define \mbrmub\ to be the subset of \mbrq\ with the constraint that the set of unitaries $\{U_b\}$ form a set of MUBs. 
\end{definition}

The definitions here given, together with the representations given in~\Cref{sec.representations}, allow us to detail the relationships depicted in~\Cref{fig.venn}, as a pictorial representation of~\Cref{th.venn}.

\begin{theorem}\label{th.venn}
    The relationship among the MBR subclasses and other representations are the following:
    \begin{align}
    \stab & \subsetneq \mbrc \\ 
    \mps & \subsetneq \mbrc \\ 
    \mbrmub & \cap \stab \neq \emptyset \\
    \mbrmub & \cap \mbrc \neq \emptyset \\
    \mbrmub & \cap (\mbrq\cap\mbrc) = \emptyset 
\end{align}
\end{theorem}
The proof can be found in~\Cref{app.relations}.

\begin{figure}
    \centering
    \includegraphics[width=\linewidth]{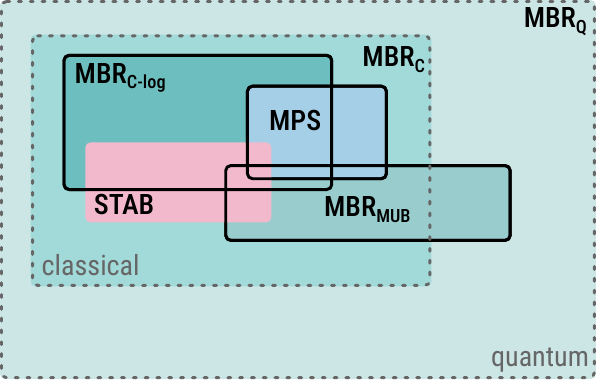}
    \caption{Known inclusion relationships of MBR and its subclasses with respect to other classical representations of quantum states. This diagram is a graphical description of~\Cref{th.venn}.}
    \label{fig.venn}
\end{figure}

\subsection{Optimality of \texorpdfstring{ \mbrmub}{MBR-MUB}}\label{sec.mbrmub}

In this section we explore certain properties that are desirable when choosing the bases that will be used for the MBR, and argue that \mbrmub\ is optimal with respect to these properties. 
Intuitively, we would like each basis to capture a different fraction of the quantum state of interest, that is, 
\begin{equation}\label{eq.orthogonality_condition}
\bra{\psi_a}U_a^\dagger U_b \ket{\psi_b} \approx 0.
\end{equation} 
The reason is that if two bases address the same components of the state, then we encounter redundancy.
For a second condition, we require the multiple bases to be necessary, i.e., it is not possible to mimic the effect of a new basis with a modest increment of $K$. We focus on the bases, therefore we give freedom to choose $S_b$ and coefficients $\alpha_b, \bm c^{(b)}$ specifying $\ket{\psi_b}$. 
The concept of overlaps among bases hinted in~\Cref{eq.orthogonality_condition} is then formalized as follows. 
\begin{definition}[$(U, K)$-fidelity]\label{def.overlap}
Let $U$ be a unitary operation specifying a change of basis. The states $\ket{\psi_1}$ and $\ket{\psi_2}$ are any two with $CB$-rank upper bounded by $K$. 
The $(U, K)$-fidelity is given by
\begin{equation}
\mathcal{F}_{K, U} = \max_{\substack{\ket{\psi_1}, \ket{\psi_2} \\ CB\left(\ket{\psi_1}\right) \leq K \\ CB\left(\ket{\psi_2}\right) \leq K}} \vert\bra{\psi_1}U\ket{\psi_2} \vert.
\end{equation}
\end{definition}

We discuss now the conditions under which $\mathcal F_{(U, K)} = 0$. 
Consider first two arbitrary states with supports $S_1, S_2$ in the computational bases. 
The changes of bases $U$ that satisfies this condition are those that fulfill $U \ket i = \ket j$, for $i \in S_{1}, j\notin S_2$. 
That is, the supports are moved to orthogonal subspaces. 
Certain permutations meet this requirement. 
However, MBRs where $U_b$ are permutations are trivially reduced to a single-basis representation with larger $K$. 
Another allowed kind of basis changes are block diagonalizable operations $U$. 
For illustrative purposes, we consider two extremal cases, (a) $U$ is diagonalizable in $\mathcal O(\exp(n))$ blocks of $\Opoly n$ size, (b) $U$ is diagonalizable in $\Opoly n$ blocks of $\mathcal O(\exp(n))$ size. 
In the case (a), the MBR can be substituted by a single-basis representation, with $\Opoly n$ increment in $K$. 
The exponential scaling is never seen, since there are \textit{only} polynomially many elements. 
On the contrary, the case (b) leads to MBRs that cannot be reduced to singular bases. 
However, this implies that only a small number of supports will be capable of satisfying~\Cref{eq.orthogonality_condition}. 
This reasoning provides the intuition that: if $\bra{\psi_1} U \ket{\psi_2} = 0$, then it is likely that the MBR can be substituted by single-basis representation with a moderate increase in the \cbe-rank.

We formulate the following statement inspired by the previous reasoning, hinting for a tension between reducing $\mathcal F_{(U, K)}$ and maintaining the condition in~\Cref{eq.orthogonality_condition}.
\begin{lemma}[Optimal choice of bases]\label{le.mubs}
Let $\mathcal{F}_{K, U} $ be the basis-dependent sparse-fidelity. 
Let $\mathcal{F}_{K}$ be 
\begin{equation}
    \mathcal{F}_{K} = \min_U \mathcal{F}_{K, U},
\end{equation}
for $U$ unitary. 
$\mathcal{F}_{K}$ is minimized if $U$ satisfies the condition
\begin{equation}\label{eq.mubs}
\vert \bra i U \ket j\vert^2 = \frac{1}{2^n}\; \forall (i, j), 
\end{equation}
with
\begin{equation}
\mathcal{F}_{K} \leq \frac{K}{\sqrt{2^n}}.
\end{equation}
\end{lemma}
The proof can be found in~\Cref{app.mubs}.
The condition given by~\Cref{le.mubs} is known as mutually unbiased bases (MUB), see~\Cref{sec.mubs}.

\subsubsection*{Properties of MBR with MUBs}

The conditions from~\Cref{le.mubs} allows us to distill some properties of the \mbrmub\ of a quantum states. 
All subsequent properties can be proven using the following following observation. 
When $\{U_b\}$ form a set of MUBs, then the low-\cbe-rank state corresponding to each bases has an exponentially vanishing overlap with all other bases. 
This means that new basis in the MBR adds an exponentially small --- but non zero --- redundancy in the previous steps. 

\begin{lemma}[Accuracy of a MBR]\label{le.bounds_error}
Let $\ket{\psi_{B, K}}$ be the MBR representation of a state $\ket\psi$, with bases $\{ U_b\}$, and let $\ket{\psi_{b}}$ be the optimal approximation of $U_b^\dagger \ket\psi$ with $CB$-rank $K$. 
The bases $\{ U_b\}_b$ form a set of MUBs. 
The overlap achievable by the MBR is bounded by
\begin{align}
\max_{\{ \alpha_b\}} \left\vert \braket{\psi_{B, K}}{\psi} \right\vert & \leq \sqrt{\sum_b \left\vert \braket{\psi_{b, K}}{\psi}\right\vert^2} \left( 1  + \frac{K B }{2 \sqrt{2^n}}\right) \\ 
\max_{\{ \alpha_b\}} \left\vert \braket{\psi_{B, K}}{\psi} \right\vert & \geq \sqrt{\sum_b \left\vert \braket{\psi_{b, K}}{\psi}\right\vert^2} \left( 1  - \frac{K B }{2 \sqrt {2^n}}\right),
\end{align}
and it is obtained by choosing $\alpha_b \propto \vert \braket{\psi_{b}}{\psi} \vert$, with proper normalization.
\end{lemma}
The proof can be found in~\Cref{app.bounds_error}.

We can also study the dimensionality of the space spanned by the MBR. 
\begin{lemma}[Dimensionality]\label{le.dimensionality}
Let $\{U_b\}$ be a set of unitary operations that transforms the computational basis into other bases, such that those bases are pairwise MUBs, and let $\{S_b \}$ be a set of sets of states in the computational basis, with $\vert S_b \vert = K$. 
It holds that the linear space spanned by $$\{ U_b \ket{i_b}\}_{\substack{i_b \in S_b \\ b \in \{1, \ldots B\}}}$$ has dimension $K B$ as long as
\begin{equation}
B \leq \frac{\sqrt{2^n}}{K} + 1.
\end{equation}
\end{lemma}
The proof can be found in~\Cref{app.dimensionality}.

Finally, we address first the difference in the volume of states that can be approximated in the single- and multiple-basis representations. 

As an auxiliary result, we begin by computing the relative volume of states that can be approximated by a $K$ sparse state with fixed support. 

\begin{lemma}\label{le.volume_sparse}
Let $\ket\psi$ be an $n$-qubit random quantum state drawn from the Haar distribution, and let $\ket{\phi_S}$ be any quantum state of the form
\begin{equation}
\ket{\phi_S} = \sum_{i \in S} c_i \ket i,
\end{equation} 
where $S$ is a fixed set of size $K \in \Opoly n$. 
Then
\begin{multline}\label{eq.probabilities}
\operatorname{Prob}\left(\max_{\ket{\phi_S}}\left(\vert\braket{\psi
}{\phi_S}\vert^2 \right) \geq 1 - \epsilon \right) = \\
I\left(\epsilon; 2^n - K,  K\right),
\end{multline}
with $I(x; a,b)$ being the incomplete regularized beta function.
\end{lemma}
The proof can be found in~\Cref{app.volume_sparse}, and a detailed description of $I(\cdot; \cdot, \cdot)$ is available at~\Cref{app.beta}.
The idea behind this proof is as follows. A quantum state can be mapped to a sphere of dimension $2^{n+1}$, and a quantum state with \cbe-rank $K$ to another sphere of dimension $2K$. 
The calculation consists in computing the volumen of a $\epsilon$-ball of dimension $2^{n + 1} - K$ around any state with \cbe-rank $K$, and normalizing appropriately. For the interest of the paper, we refer the approximation
\begin{multline}
    I\left(\epsilon; 2^n - K - \frac{1}{2}, K - \frac{1}{2}\right) \approx \\ \Phi\left(\epsilon; \frac{2^n - K}{2^n}, \frac{(2^n - K)(K)}{(2^n + 1)2^{2n}} \right), 
\end{multline}
with $\Phi(\cdot; \mu, \sigma^2)$ being the cumulative distribution function of a Gaussian distribution with mean $\mu$ and variance $\sigma^2$~\cite{perez-salinas2024analyzing}.

For completeness, we provide the intuitively more comprehensible bound, valid for $K \leq (1-\epsilon)2^n$,
\begin{multline}
\operatorname{Prob}\left(\max_{\ket{\phi_S}}\left(\vert\braket{\psi
}{\phi_S}\vert^2 \right) \geq 1 - \epsilon \right) \in \\
\mathcal O \left(\exp\left( - 2^n\left( \epsilon - \frac{2^n - K}{2^n}\right)^2 \right)\right).
\end{multline}
This bound comes from connecting the function $I(x; a, b)$ to the binomial distribution and applying Hoeffding's inequality to such binomial distribution. 
We leverage the result  to compute the fraction of states that can be represented within $\epsilon$ accuracy using the MBR. 
\begin{corollary}\label{cor.volume_mbr}
Let $\ket\psi$ be an $n$-qubit random quantum state drawn from the Haar distribution, and let $\ket{\psi_{K, B}}$ be its optimal MBR. 
The $B$ bases are chosen from a set of $M$ different MUBs. 
Then, for $K \leq (1 - \epsilon) 2^n$, 
\begin{multline}
\operatorname{Prob}\left(\vert\braket{\psi
}{\psi_{K, B}}\vert^2 \geq 1 - \epsilon \right) \in \\ \mathcal{O}\left(\exp\left( - 2^n\left( \epsilon - \frac{2^n - K}{2^n}\right)^2 \right) C_M(K, B) \right).
\end{multline}
with 
\begin{equation}\label{eq.Ckb}
C_M(K, B) = \binom{M}{B}\binom{2^n}{K}^B
\end{equation}
\end{corollary}
The proof follows immediately from~\Cref{le.volume_sparse} together with two considerations. 
First, the MBR spans in a linear subspace with dimension $KB$. 
Second, there are $M$ different MUBs, from which we choose $B$ of them. 
However, the $M$ bases do not count any ordering, and therefore there are $\binom{2^n}{K}$ different choices for $K$-sparse states for each basis. 
The overlaps among volumes are absorbed in the upper bound. 

It is interesting to compare the bounds in the probabilities given by~\Cref{cor.volume_mbr} for extremal cases, in which $K = 1$ or $B= 1$ while keeping the same number of degrees of freedom. 
We obtain the quotients
\begin{align}\small
\frac{C_M(K, B)}{C_M(KB, 1)} & \leq M^{B - 1}\frac{(KB)^{KB}}{B! (K!)^B} \\
 \frac{C_M(K, B)}{C_M(KB, 1)} & \geq M^{B - 1} \frac{(KB)!}{(B K^K)^B} \\
\frac{C_M(K, B)}{C_M(1, KB)} & \leq  \frac{M^B}{M^{K - 1}}\frac{(KB)^{KB}}{B! (K!)^B} \\ 
\frac{C_M(K, B)}{C_M(1, KB)} & \geq \frac{M^B}{M^{K - 1}}\frac{(KB)!}{(B K^K)^B} 
\end{align}
These numbers indicate that the volume of representable states increases if the number of considered bases $B$ increases.

We give an interpretation of the results obtained in this section. The bounds given by~\Cref{le.volume_sparse} and subsequent approximations state that the volume of captured states rapidly increases at the transition point $\epsilon \approx 2^{-n}(2^n - K)$, yielding exponentially vanishing relative volumes for reasonable $K \in \Opoly n, 1 - \epsilon \in \Omega({\rm poly}^{-1} (n))$. This observation aligns with the common observation that no classical representation of states can reach all (or a significant fraction) of quantum states. This, however, does not contradict the results on classical (in)tractability detailed in this manuscript.

\section{Applications of MBR}\label{sec.applications}

Succinct representations of quantum states acquire importance depending on their utility to conduct certain tasks. 
Tensor networks, and in particular MPS, have been broadly utilized to investigate ground states of many-body physics. 
The stabilizer formalism found great applicability in quantum error correction~\cite{terhal2015quantum}. 
In order to highlight the potential of MBR, we outline in this section three different applications to it.

\subsection{Ground-state approximation via MBR}\label{sec.groundstate}

The MBR can be used as a tool to obtain classical approximations of ground states of Hamiltonians, by leveraging quantum subspace expansions~\cite{mcclean2017hybrid}. Notice that finding the ground states even for physically motivated Hamiltonians is QMA-complete~\cite{schuch2009computational}, hence proposals for heuristic approaches becomes practically relevant. 
In the case we deal with \mbrc\, then this approximation is also classically efficient to compute. 

We assume a decomposition of the Hamiltonian as 
\begin{equation}\label{eq.ham_decomposition}
    \hat H = \sum_c \hat h_c, 
\end{equation}
such that each term $\hat h_c$ can be diagonalized in a basis defined by $U_c$, that is
\begin{equation}
\hat h_c U_c\ket i = \lambda_i U_c\ket i, 
\end{equation}
for $i \in \{0, \ldots 2^n - 1 \}$. 
Starting from this point, we want to obtain the best possible MBR approximation to the ground state. 
For that, we consider $\{ U_c \}$ as our set of unitary operations, and develop a constructive method to find the coefficients $\bm c^{(b)}, \alpha_b$, as well as the sets $S_b$. 

We assume first that $S_b$ is fixed. 
Then, the Hamiltonian as evaluated in the considered subspace is given by
\begin{equation}
\hat H_{i_b, j_a} = \sum_c \bra{i_b} U_b^\dagger \hat h_c U_a \ket{j_a}, 
\end{equation}
where the indices $\{a,b,c\}$ run over the same values, and $\{i_b, j_a\}$ are given by the subsets $S_b$ from the MBR. 
Since this truncated Hamiltonian has $\Opoly n$ size, the next step is just to diagonalize. 
However, the basis in which this Hamiltonian is expressed is not orthonormal due to the presence of the unitary matrices $U_b$, and the relationship between basis elements is encoded in the matrix $F$ from~\Cref{eq.fidelity_matrix}. 
The exact diagonalization can be performed through the generalized eigenvalue problem, where the solution state lives in the truncated space explored by the MBR, that is solving the equation~\footnote{Notice the small abuse of notation by writing in the previous equations the matrices with the same symbol we used to define their components.}
\begin{align}
\det\left(\hat H_{i_b, j_a} - E F_{i_b, j_a}\right) = 0.
\end{align}
Assuming a polynomial \mbrc, solving the generalized eigenvalue problem is classically feasible if $H$ decomposes in polynomially many terms $\hat h_c$. 
On the contrary, assuming \mbrq, both matrices $\hat H_{i_b, j_a}$ and $F_{i_b, j_a}$ must be obtained with the assistance of a quantum computer via Hadamard tests or similar techniques for overlap estimation, while the diagonalization is a classical task. 
This diagonalization provides exactly the ground state if $\det F = 0$, that is the MBR, for given $K$ and unitaries, can represent any quantum state. 
This approach yields an upper bound on the energy of the approximation of the ground state of interest captured by our ansatz.
This scenario is only possible for few qubits, and MBR is not needed to approximate the ground state.

We address now the question of finding the subspaces in which the approximation to the ground state is optimal, that is the choice of $S_b$. 
Searching among all possibilities is a hard problem if no further information is provided. 
There are $\binom{2^n}{K}$ different sets to choose from, and any brute-force search becomes unfeasible. 
Only informed searches provide tractable scenarios. 
As a first candidate, one can use heuristic optimization methods to look for sufficiently good $\{ S_b\}$, e. g. via genetic algorithms~\cite{mitchell1996introduction}. 
One can also choose a physics-informed alternative by selecting the lowest-energy states in each of the Hamiltonians $H_b$, provided this information is efficiently computable with classical resources.

\subsubsection*{Transverse-Field Ising model}

We illustrate the method here described with the example of approximating the ground state of a transverse-field Ising model (TFIM), defined as
\begin{align}\label{eq.tfim}
\begin{aligned}
H_G(J, h) &= J\sum_{(i, j) \in E(G)} X_i X_j + h \sum_{i \in N(G)} Z_i \\ 
&\equiv J H_{xx} + h H_z,
\end{aligned}
\end{align}
where $G$ is a graph defining the model and $N(G), E(G)$ are respectively nodes and edges of the graph.

We first choose the bases to apply the MBR. 
In this case, the bases can be chosen as those diagonalizing $H_{zz}$ and $H_z$ respectively. 
$H_{z}$ is diagonalized in the computational basis, thus $U_{z} = I$, and $H_{xx}$ is diagonalized in the Hadamard basis, thus $U_{xx} = H^{\otimes n}$. 

For the subsets $\{ S_b\}$, we follow the physical intuition in the considered problem. 
Since we look for a ground state, we apply the heuristic of keeping those states with the lowest energies in each term $\hat h_c$. 
For the Hamiltonians considered, the energies of a basis element defined by a bitstring can be easily computed. 
For $H_z$, one just needs to obtain the Hamming weight. 
For $H_{xx}$, the energy is related to the number of cuts of the graph $G$ induced by the bitstring, in the Hadamard basis. 
Depending on the value of $J$, the heuristics are different. 
For the ferromagnetic phase $J < 0$, the ground state of $H_{zz}$ is equivalent to solving the MinCut problem, which has a polynomial cost~\cite{goldschmidt1994polynomial}. 
The antiferromagnetic phase $J > 0$ needs the heuristic to solve a MaxCut problem, which is NP-hard for arbitrary graphs Heuristic methods exist to find approximations to the solutions in ~\cite{goemans1995improved, dunning2018what}, and each approximation found by any method can be used as one element in $H_{zz}$. 
In some physics-motivated cases, such as lattice graphs, the MaxCut problem is trivial to solve.

\begin{figure}[t!]
\includegraphics[width = \columnwidth]{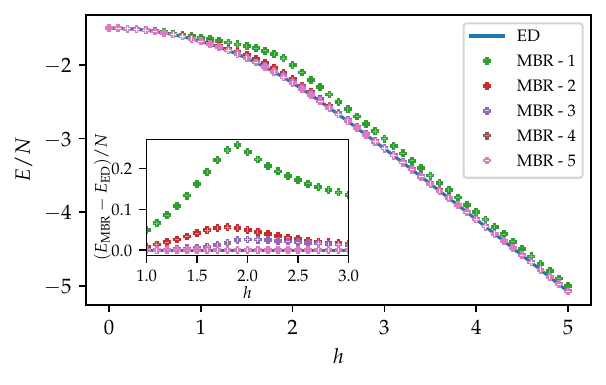}
\caption{Different MBR approximations of the ground state of an anti-ferromagnetic Ising model. 
The system is defined on a square lattice of size $4\times 4$ with open boundary conditions, and $J = 1$.
The inset shows the deviation from exact diagonalization data with increasing degree of the MBR approximation in the middle of the parameter range, i.e. the hardest regime for the MBR approximation.}
\label{fig.energies}
\includegraphics[width= \columnwidth]{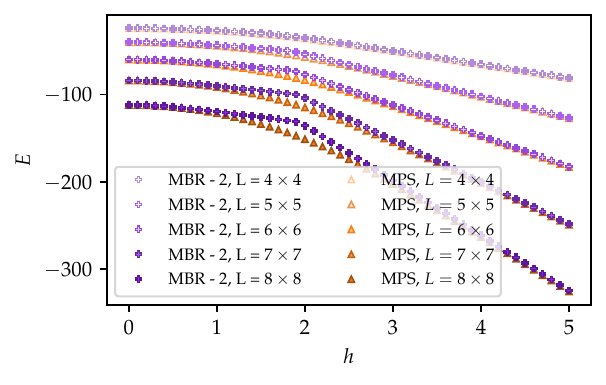}
    \caption{Finite size effects in the approximation for the TFIM on an $L\times L$ lattice.
    As the system grows, the deviations around the phase transition between MPS data and MBR data become more pronounced. 
    Larger systems can need a larger MBR basis.}
\label{fig:finite-size}
\end{figure}

We numerically approximate the ground state of the TFIM [cf.~\Cref{eq.tfim}] on a 2D square lattice of size $L \times L$. 
The considered space includes the ground states of $H_{xx}$ and $H_z$ in the computational and Hadamard basis respectively, and all states up to a fixed Hamming distance $D_H$ to these ground states. 
The dimensionality of the expanded subspace grows with the system size as $\mathcal{O}\left(\left(L^{2}/H_W\right)^{D_H}\right)$. 
The regularity of the lattice allows for the efficient computation of these state.
The accuracy is assessed through the comparison between the energies of found states through MBR, and that obtained with exact diagonalization or MPS approximation, depending on the size of the lattice.

First, we benchmark the MBR state against exact results (cf.~\Cref{fig.energies}) on a small lattice in the anti-ferromagnetic case ($J = 1$) with varying transverse magnetic field  $h$.
The degree of the MBR approximation is given by the dimension of the considered subspace, here depicted as ${\rm MBR}-D_H$. 
The energies of the MBR approximations to the ground state energy match the values obtained by exact diagonalization increasingly well as $D_H$ increases (cf. inset of~\Cref{fig.energies}).
The MBR approximation is capable of accurately representing the regimes $h \approx 0$ and $h\gg J$ even with of $D_H = 1$. 
We attribute this to the choice of the subspace expansion which includes the ground state of each term in the Hamiltonian. 
In both regimes, high and low transverse field, the ground state becomes more similar to a state in our ansatz. 

For the regime around the phase transition $3 J \approx h$~\cite{blote2002cluster, ducroodejongh1998critical}, the MBR approximation yields less accurate results as compared to the reference value. 
The accuracy improves as the degree of the MBR is increased, fulfilling the interpretation that the approximation will become a exact representation with enough computational power. 
In the MBR formulation, the deviation from the MPS simulation happens slightly earlier than the expected formulation.
This could be related to a breakdown of the heuristic in this parameter regime, even independent of the actual phase transition.

In~\Cref{fig:finite-size}, we compare the performance of a fixed MBR approximation for increasing system size.
While a small MBR degree is sufficient to obtain good estimates for the ground state energy on small lattice sizes, a larger degree becomes necessary when the lattice increases.

It is interesting to analyze the differences between the methods using MBR and \mps\ to obtain ground states (cf.~\Cref{fig:finite-size}) 
MBR is entirely based on heuristics over the subspace to consider in the search for the ground state. 
Candidates to ground states are accumulated, without adding correlations to the selection criteria. 
On the other hand, \mps\ are highly dependent on correlations of the ground state to be found, and are successful only in the case of low entanglement. 
This hinders the applicability of \mps\ for systems allowing long-distance correlations, such as 2D lattices, or lattices defined by arbitrary and irregular graphs. 
Therefore, the candidates to ground states given by MBR might be more accurate in those cases, where MPS demand large overheads. 
This phenomenon can extend to quantum circuits with arbitrary connectivity which do not admit classical simulations. 
Investigations on these topics is left for future research.

\subsection{MBR and \ryc}\label{sec.reducechop}

We inspect now the interplay between MBR and the \ryc\ algorithm~\cite{perez-salinas2023shallow} to show how using several bases might be useful to further decrease the depth requirement.

\subsubsection{Background}

The \ryc\ method leverages states with small \cbe-ranks to conduct simulations of deep quantum computations making use of shallower circuits~\cite{perez-salinas2023shallow}. 
Consider the problem of computing $p(x) = \vert \bra x \mathcal U \ket 0\vert^2$, where $x$ is a given bitstring and $U$ is a given (deep) quantum circuit. 
We can approximate $p(x)$ by
\begin{multline}\label{eq.reduceandchop}
p(x) = \left\vert \bra x \mathcal U \ket 0 \right\vert^2 = \left\vert \sum_{i = 0}^{2^n - 1} \bra x \mathcal U_2 \ket i \bra i 
\mathcal U_1 \ket 0\right\vert^2 \\ \approx \left\vert \sum_{\substack{i \in S \\ \vert S \vert = K}} \bra x \mathcal U_2 R^\dagger \ket i \bra i R \mathcal U_1 \ket 0\right\vert^2,
\end{multline}
where $S$ $K$ is the support over bitstrings $i$, $\mathcal U = \mathcal U_2 \mathcal U_1$ and $R$ is a unitary operation called the reducer. 
For the approximation to be accurate, it is sufficient that the state $R \mathcal U_1 \ket 0$ fulfills $\cbe(\ket\psi)\in \Opoly n$, with $n$ the number of qubits. 
Under this assumption, the \ryc\ method provides a recipe to approximate $p(x)$ with classical overhead in $\Opoly n$. 
The depth requirements to compute $p(x)$ decrease from $d(U)$ to $d(U)/2 + d(R)$ through \ryc (assuming $\mathcal U$ is splitted in halves), with $d(\cdot)$ being the depth of a quantum circuit.

A fundamental step for \ryc\ to work is to find the reducer $R$. 
The function of $R$ is to identify a change of basis under which the state $\mathcal U_1 \ket 0$ has small \cbe-rank. 
The hardness of finding the appropriate $R$ hinges on the internal structure of state $\mathcal U_1 \ket 0$. 
For example, there exist straightforward circuits with constant depth with the capability to transform a state with maximal CB rank into another state with minimal rank.  
The search for $R$ can be approached variationally~\cite{perez-salinas2023shallow}, and a successful search for $R$ dramatically affects the performance of \ryc. 
Recent advances in peaked states~\citep{aaronson2024verifiable}, which can be understood as states with \cbe-rank 1, showed that it is possible to peak a random quantum circuit $d(V)$ by adding a reducer satisfying
\begin{equation}
d(R) \in \, \Omega\left( \left(\frac{d(U)}{n} \right)^{0.19}\right).
\end{equation}
in the same domain, it was shown that $O(1)$-depth peaked circuits can be classically and optimally simulated in quasi-polynomial runtime $n^{\bigO\left(\log(n)\right)}$~\cite{bravyi2023classical}, and therefore \ryc\ can bring quantum advantages starting from $d(R \mathcal U_1) \in \Omega(1)$.

\subsubsection{MBR for depth reduction}

The \ryc\ method relies on identifying an easily implementable change of basis $R(\theta)$ to achieve a small $\cbe$-rank for a given state. 
The MBR offers a more nuanced representation of a quantum state by employing multiple bases simultaneously, each representing a distinct component of the state of interest. 
This representation further explores the depth reduction motivating \ryc. 
The MBR is interesting as long as the bases required for it are easy to implement in a quantum computer. 
Ideally, these bases are designed to be even simpler than the original $R(\theta)$ from \ryc, striking a balance between $\cbe$-ranks, number of bases, and the associated cost of basis change.

To clarify this, consider a state $\ket\psi$ and a poly-sized set of MUBs $\{U_b\}$ such that $(\cbe(U_b\psi)), (1 - \epsilon)^{-1} \in \bigO({\rm poly}(n))$. 
By virtue of~\Cref{le.bounds_error}, the accuracy of the MBR can reach constant values. 
The action of $R$ needed to obtain sparse representation is now split among several more simple unitaries. 
The complexity of running each circuit is therefore simpler. 

Estimating the depth savings attained from combining the MBR with the \ryc\ method is more complicated. 
As motivated by the example in~\Cref{sec.example}, it is possible in some cases to find bases represented by shallow circuits for which an accurate MBR of the state is obtained, while the basis required to lower the \cbe-rank of the entire state requires larger depth. 
In the \ryc\ method, the output of the first half of the circuit must admit an accurate MBR, for $K, B \in \Opoly n$. 
The results in this section bound the depth needed to attain such states, thus allowing us the to substitute this circuit with classical combinations of quantum circuits of depth at most $d(U)$. 
If this state serves as input for the \ryc\ algorithm, then the required depth in a quantum circuit to simulate the values of $p(x) = \vert \bra x U \ket{\psi_{K, B}}\vert^2$ is upper bounded by $d(U) + \max_b d(U_b)$, yielded from concatenating the change of basis and the circuit $U$.

\begin{figure*}
\includegraphics[scale=1.5]{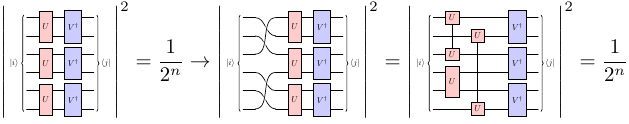}
\caption{Let $U$, $V$ be two $k$-qubit MUBs, such that its tensor products $U^{\otimes n/k}, V^{\otimes n/k}$ are MUBs for $n$ qubits. 
Since MUBs do not have ordering, it is possible to construct new sets of MUBs by applying arbitrary permutations to the bases $U^{\otimes n/k} \rightarrow U^{\otimes n/k} \Pi$. 
This allows for finding many different sets of MUBs.}
\label{fig.mub-perm}
\end{figure*}

The first result we give is related to lower bounds to create arbitrary $n$-qubit states with $CB$-rank $K$. 
\begin{lemma}[Lower-bound in creating sparse states]\label{le.lowerboundssparse}
Arbitrary $n$-qubit states with CB-rank $K$ require at least $m$ gates in depth $d$ to be prepared in a quantum computer from the $\ket 0$ state, with
\begin{align}
m \in & \Omega\left( \frac{K \log(\epsilon^{-1})}{\log\log K} + n \right) \\ 
d \in & \Omega\left( \frac{K \log(\epsilon^{-1})}{\log K \log\log K} + n \right)
\end{align}
\end{lemma}
The proof immediately follows from Section 4.5.5 in~\cite{nielsen2010quantum} to construct states of dimension $K$ in a register of $\log K$ qubits. 
The additional term $n$ accounts for arbitrary correlations among qubits. 

For the next step, we recall that any state in \mbrmub\ has dimension at most $KB$, assuming no redundancy in the bases. 
Therefore, we can further leverage~\Cref{le.lowerboundssparse} to obtain the following result. 
\begin{corollary}[Lower-bound in MBR]\label{le.lowerboundsmbr}
Let $\ket\psi$ be an $n$-qubit such that it can be $\epsilon$-approximated with an \mbrmub\ as in~\Cref{eq.state_def}, with sparsity $K$ and $B$ different bases. 
This state requires at least $m$ gates and depth $d$ to be created, with
\begin{align}
m \in & \Omega\left( \frac{K B \log(\epsilon^{-1})}{\log\log (K B)} + n \right), \\ 
d \in & \Omega\left( \frac{K B \log(\epsilon^{-1})}{\log (K B) \log\log (K B)} + n \right).
\end{align}
\end{corollary}
The proof uses a modification of the proof in Section 4.5.5 in~\cite{nielsen2010quantum}, which does not rely on any choice of basis, but only on counting dimensions. 
This result holds but is less tight for arbitrary MBR.

We give now upper bounds by providing a constructive algorithm to prepare MBR states into a quantum computer of the form of~\Cref{eq.state_def}. 
This algorithm assumes availability of preparation of states with $CB$-rank $K$ and at circuit specifications of the bases $U_b$. 
Note that black-box query to $U_b$ is enough to prepare $\ket{\psi_{K, B}}$, but descriptions are needed to classically compute relative overlaps among states. 
The algorithm we propose utilizes the linear combination of unitaries as a subroutine~\cite{dewolf2023quantum}. 
The first step is to create a state with a non-zero overlap with the desired state $\ket\psi$ in the MBR. 
\begin{lemma}\label{le.lcu}
Let $\ket\psi$ be an MBR of a state with sparsity $K$ and number of basis $B$, where the weight for each basis is given by $\bm\alpha$, as in~\Cref{eq.state_def}. 
Let $\ket\Psi$ be the output of a quantum state created as a linear combination of unitaries. 
The success probability is bounded by
\begin{equation}
\vert \braket{\Psi}{\psi_{K, B}}\vert^2 \leq \frac{\Vert \bm \alpha \Vert_2^2}{\Vert \bm \alpha \Vert_1^2}\left( 1 - \max_{a\neq b} \Vert F_{a, b}\Vert_\infty \right) + \maxF,
\end{equation}
where $F_{a,b}$ is the block of the overlap matrix in~\Cref{eq.fidelity_matrix} corresponding to the bases $(a, b)$. 
\end{lemma}
The proof can be found in~\Cref{app.lcu}. 

We can leverage the previous result to obtain upper and lower bounds to create the state $\ket\psi$ with arbitrary precision, using the recipe from the previous paragraph. 

\begin{lemma}[Upper-bounds to prepare MBR states on a quantum computer]\label{le.upperboundsmbr}
Let $\ket\psi$ be the MBR of a state with sparsity $K$ and number of basis $B$, with weights $\alpha$, as in~\Cref{eq.state_def}. 
Let $d(K)$ be the required depth to prepare a $K$-sparse state in the computational basis, and let $d(U) = \max_b d(U_b)$ be the maximum depth required for the basis change. 
The depth $D$ required to load $\ket\psi$ with success probability at least $1 - \delta$ is given by
\begin{multline}
D \leq d(W) + 2 \ (d(U) + d(K)) \sqrt{B} \log(2/\delta)\frac{\Vert \bm\alpha\Vert_1^2}{\Vert \bm\alpha\Vert^2_2}\\  \left(1 - \frac{\maxF}{2}\left( \frac{\Vert \bm\alpha\Vert_1^2}{\Vert \bm\alpha\Vert_2^2} + 1\right) \right).
\end{multline}
\end{lemma}
The proof can be found in~\Cref{app.upperboundsmbr}.

The previous bounds purposely did not specify $d(K)$ and $d(U)$. 
First, we do not impose yet any constraint on the depth needed to apply the changes of bases. 
These will depend on the bases themselves, and choosing these bases is a challenge on its own. 
Second, the number $d(K)$ strongly depends on the number of ancillary qubits available~\cite{zhang2021lowdepth, sun2023asymptotically, ramacciotti2023simple}. 
These bounds become tight and imply an approximate lower-bound under the assumption that the only available tools are the preparation of sparse states in the CB and the basis unitaries $U_b$ separately. 
This scenario acquires relevance in the context that all states for each basis are peaked~\cite{bravyi2023classical}. 
Finally, we use ancillary qubits to prepare the control register $\ket b$. 
The use of ancillary qubits is necessary to create this superposition in the aforementioned scenario~\cite{oszmaniec2016creation}.

Additionally, we can adapt the results from~\Cref{le.lcu} and~\Cref{le.upperboundsmbr} to obtain new upper bounds in MBR-state preparation in the case the bases form sets of MUBs. 
\begin{corollary}\label{cor.lcu}
Given the conditions of~\Cref{le.lcu}, for \mbrmub, then
\begin{equation}
P_{\rm success} \leq B \frac{\Vert \bm \alpha \Vert_\infty^2}{\Vert \bm \alpha \Vert_1^2}\left( 1 - \frac{K \ B}{\sqrt{2^n}} \right) + \frac{K \ B}{\sqrt{2^n}}.
\end{equation}
\end{corollary}The proof is analogous to that of~\Cref{le.lcu}, usingalso~\Cref{le.mubs}.

\subsubsection*{Combinatorial considerations}\label{sec.klocal}

Finally, we restrict the depth of the unitary operations. 
In order to maintain the MUB condition, we focus in the case in which the MUBs are local in subsystems of $l$ qubits, with $l \in \mathcal O(1)$. 
This example illustrates that interesting MBRs are accessible with minimal resources. 

Consider now the example of complete sets of MUBs in subsystems of size $l$, with $M = 2^l + 1$. 
We can compute the asymptotic scaling in~\Cref{eq.Ckb} as
\begin{equation}
C_M(K, B) \in \exp\left(\tilde{\mathcal O}\left({\rm poly}(n) \right)\right), 
\end{equation}
assuming $K, B \in \Opoly n$. 
If we restrict the permutations to qubits, there are several choices to be done on how to select the $(n/l)$ partitions of size $l$. 
In addition, one can choose several different sets of MUBs. 
This leads to large numbers of possible combinations 
\begin{align}
\sharp_{\rm MUB\ sets}(n, l) & = \left((2^l + 1)!\right)^{n/l} \\ 
\sharp_{\rm Partitions}(n, l) & = \frac{n!}{(l!)^{n/l} (n/l)!} \\  
\sharp_{\rm Partitions}(n, l) \sharp_{\rm MUB\ sets}(n, l) & \in  e^{\tilde{\mathcal O}\left(n 2^l\right)}.
\end{align}

We dive now into the numerous possibilities opened by the $k$-qubit MUBs. 
We first note that the MUB condition $\vert \bra i U_2^\dagger U_1 \ket j\vert = 2^{-l/2}$ holds under permutations $\Pi$ over the elements $\ket i, \ket j$. 
That is, if $U_1, U_2$ are MUB, then all bases $U_1 \Pi_1, U_2 \Pi_2$ are MUBs, for any pair of permutations $\Pi_1, \Pi_2$. 
The permutations can be understood as reorderings of the elements in the computational basis. 
Thus, if we choose two $l$-qubit MUBs to be simultaneously applied in a $n$-qubit system, then all permutations of basis elements allow for the MUB condition as well. 
See~\Cref{fig.mub-perm} for a graphical description of qubit permutation. 
This argument increases the number of sets of MUBs, but not the size of each set. 
If we consider $l$-qubit subsystems in a $n$-qubit state, without permutations, the number of different sets grows as
\begin{equation}
\sharp_{\rm MUB\ sets}(n, l) = \left((2^l + 1)!\right)^{n/l} \in \tilde\Theta\left(e^{n 2^l}\right),
\end{equation}
where the scaling is obtained following Stirling's approximation.
The number of different size-$l$ partitions of a set of size $n$ grows as
\begin{multline}
\sharp_{\rm Partitions}(n, l) = \frac{n!}{l!^{n/l} (n/l)!}  \\ \in \tilde\Theta\left( \left(\frac{n}{l}\right)^{n\frac{l - 1}{l}}e^{n/l}\right),
\end{multline}
yielding a total number of combinations
\begin{equation}
\sharp_{\rm Partitions}(n, l) \sharp_{\rm MUB\ sets}(n, l) \in \\  \tilde{\mathcal O}\left( \left(n e^{2^l}\right)^n\right).
\end{equation}
Imposing the condition $l \in \mathcal O(\log n)$ in this reasoning, the number of different available sets of MUBs scales as
\begin{equation}
\sharp_{\rm Partitions}(n, l) \sharp_{\rm MUB\ sets}(n, l) \in \\  \tilde{\mathcal O}\left(n^n\right).
\end{equation}

Notice that it is also possible to use the $l$-qubit MUBs to construct bases that are no longer MUBs in the full $n$-qubit system. 
Such condition allows us to adapt the orthogonality relationships derived in~\Cref{sec.mubs} to the use of only some $l$-qubit bases. 
MBR as a general framework is still usable, since the basis can still be constructed as Clifford circuits, and~\Cref{eq.fidelity_matrix} is classically accessible.

\subsection{Finding multiple-basis representations of unknown states}\label{sec.algorithm_mbr}

In this section, we outline a protocol to find a MBR description of pure states. The input inputs are the quantum state $\ket\psi$ and the bases set $\{ U_b\}$, and output are the coefficients and supports of the MBR. 
Note that finding a set of bases that effectively capture a given state of interest is in general a hard problem, which will be left as future research. 
However, once the sets of bases is given, obtaining the MBR is easy, since the supports can be obtained by sampling, and the coefficients by conducting experiments with small overheads on the corresponding supports.

We are interested in finding an approximation to a quantum state $\ket\psi$  in its MBR, as in~\Cref{eq.state_def}. 
The recipe to find the MBR of an unknown quantum state is summarized in~\Cref{tab.algorithm}, and detailed in the following paragraphs.

\begin{table}
\begin{framed}

\begin{minipage}{.9\linewidth}
\begin{enumerate}[a.]
\item[\bf Input:] - Several copies of $\ket\psi$
\item[ ] - Set of bases $\{U_b\}$, error $\epsilon$
\item[\bf Output:] - Sets of coefficients $\left(\{\alpha_b\}, \{c^{(b)}_{i_b}\}\right)$, see~\Cref{eq.state_def}
\item[] - Sets of supports $\{S_b\}$, see~\Cref{eq.state_def} 
\item[\bf Steps: ] 
\item For each basis, apply the operation $U_b^\dagger$ to the quantum state $\ket\psi$.
\item For each basis, Measure the quantum state in the computational basis to identify the most relevant elements $\{ S_b\}$. 
\item Merge data from all bases and compute the fidelity matrix from~\Cref{eq.fidelity_matrix}.
\item Measure the amplitudes associated to $\{S_b\}$ for each basis. 
\item Revert the effect of non-orthogonality to obtain the coefficients $\left(\{\alpha_b\}, \{c^{(b)}_{i_b}\}\right)$. 
\end{enumerate} 
\end{minipage}
\end{framed}
\caption{Overview of the algorithm needed to obtain the MBR representation of a quantum state. 
A more detailed description of each step is available in~\Cref{sec.algorithm_mbr}, including the costs associated to each of them.}\label{tab.algorithm}
\end{table}

\paragraph{Applying $U_b^\dagger$}

The first step of the algorithm is to apply $U_b^\dagger$ to the input state. 
This unitary operations has the role of expressing the quantum state in a basis in which its \cbe-rank is small, for a moderately large $\epsilon$.

\paragraph{Sampling}

The obtained samples of a quantum state, in the limit of infinite samples, gives an approximation of the (absolute values of the) amplitudes in $\ket{\psi}$, in the sampling basis. 
In the context of MBR, sampling is not strictly necessary to obtain these amplitudes, this step is addressed later, but it is required to obtain the supports $S_b$ that define the MBR.

Assume sampling from an unknonwn state $\ket\psi$, with sampling budget $M \in \Opoly n$. 
As an extreme example, assume that $\ket\psi$ induces a uniform distribution. 
With high probability, $M$ measurements output $M$ different bitstrings, and therefore it is statistically impossible to distinguish uniform distributions over all bitstrings and only the observed ones. On the other hand, assume a peaked quantum circuit $U$ satisfying~\cite{aaronson2024verifiable, bravyi2023classical} 
\begin{equation}
    \max_{x \in \{0, 1\}^n} \vert \bra x U \ket 0\vert \in \Omega({\rm poly}^{-1}(n)).
\end{equation}
It is efficient to classically verify if $U$ is peaked just by circuit. Intuitively speaking, this tomographical protocol requires a stronger peaked-condition, meaning at least $K$ outcomes can be sampled with probability $\Omega({\rm poly}^{-1}(n))$.
We proceed now to describe the concentration properties needed for accurate sampling.

First, we define the rank of a probability distribution, in a similar way to the~\cbe-rank~\Cref{def.cbe}.
\begin{definition}[\epsrank of a probability distribution]\label{def:epsrank}
Let $i$ be a random variable that can take values $i = \{0, 1, \ldots, N - 1\}$, each with probability $p_i$, and $\sum_i p_i = 1$. 
The outcomes are ordered such that $p_i \geq p_{i + 1}$ without loss of generality. 
Let \begin{equation}
S_K(p) = \sum_{i = 1}^{K} p_i.
\end{equation}
 The \epsrank of $p$ is defined as
\begin{equation}
R_\varepsilon(p) = \min_K\left\{ K: S_K \geq 1 - \varepsilon \right\}
\end{equation}
\end{definition}
This definition allows us to formulate the following statement to give the conditions needed to identify the most relevant components of the quantum state by sampling. 
\begin{theorem}[Concentration properties]\label{th.concentration}
Let $i$ be a random variable that can take values $i = \{0, 1, \ldots, N - 1\}$, each with probability $p_i$, and $\sum_i p_i = 1$. 
The outcomes are ordered such that $p_i \geq p_{i + 1}$ without loss of generality. 
The probability distribution is sampled $M$ times, with $m_i$ samples corresponding to outcome $i$. 
For every $\epsilon$, the outcomes $m_i$ allow to estimate the \epsrank $R_\varepsilon(p)$ as $K$. 
The estimation $K$ is the only possible value compatible with the observed outcomes with probability at least $1 - \delta$ if
\begin{equation}\label{eq.peaked_constraint}
m_{K+1} \geq \sqrt M \log(2/\delta).
\end{equation}
The true value $S_K$ is estimated within precision
\begin{equation}
\Delta S_K \in \bigO\left( \frac{1}{\sqrt M}\right).
\end{equation} 
\end{theorem}
The formal proof can be found in~\Cref{app.concentration}, but a sketch is enough to comprehend the gist of this result. 
For a quantum state $\ket\psi$ we sample its outcomes in the computational basis according to some underlying probability specified by the state. 
We are interested in computing $K$, through the quantity $S_K(p)$, which can be estimated with accuracy $\mathcal O(M^{-1/2})$ via standard estimators and Monte Carlo arguments. 
We repeat this process for $S_{K+1}(p)$. 
Notice that only one sampling round is needed. 
Since we assume $p_{i} \geq p_{i + 1}$, then $S_{K+1}(p) - S_{K}(p) \geq S_{K+2}(p)- S_{K+1}(p)$. 
Due to the sampling uncertainty, we can accurately compute $K$ as long as $S_{K+1}(p) - S_{K}(p) \in \Omega(M^{-1/2})$, which forces the constraint in~\Cref{eq.peaked_constraint} and constitutes the main result in~\Cref{th.concentration}.

In the light of the algorithm to obtain MBR, the supports $\{S_b\}$ for each change of basis are obtained by just sampling the state $U_b^\dagger \ket\psi$ a total amount of $M$ times, and saving all outcomes signaling more than $\sqrt M$ times.

\paragraph{Computing amplitudes}
After identification of supports $\{ S_b\}$ is complete, we can proceed to estimate the quantities $\alpha_b c_{i_b}^{(b)}$. 
However, the only quantities that are accessible through quantum measurements are
\begin{equation}\label{eq.measurable_element}
a^{(b)}_{i_b} = \bra{i_b} U_b^\dagger \ket\psi.
\end{equation}
Both sets of quantities are related by 
\begin{align}
\begin{aligned}
\label{eq.transform_elements}
a^{(b)}_{i_b} &= \alpha^{(b)} c^{(b)}_{i_b} + \sum_{\substack{a = 1 \\ a \neq b}}^B \sum_{j \in S_a} \left(\alpha^{(a)}c^{(a)}_{j_a} \bra{i_b}U_b^\dagger U_a\ket{j_a}\right) \\
&=\alpha^{(b)} c^{(b)}_{i_b} + \sum_{\substack{a = 1 \\ a \neq b}}^B \sum_{j \in S_a} F_{i_b, j_a} \alpha^{(a)}c^{(a)}_{j_a}. 
\end{aligned}
\end{align}
For simplicity we group the elements $a^{(b)}_{i_b}$ in a vector $\bm a$. 
The comparison with~\Cref{eq.state_def} is that $a^{(b)}_{i_b} = \alpha_b c^{(b)}_{i_b}$, and the value $\alpha_b$ is obtained by imposing the normalization constraint from~\Cref{eq.normalization}. 
In a more compact fashion, we observe
\begin{equation}\label{eq.a_alpha_compact}
\bm a  = F \bm\alpha. 
\end{equation}

We first focus on obtaining the quantities given by~\Cref{eq.measurable_element}.
Such estimation requires both the real and imaginary parts of $a^{(b)}_{i_b}$. 
For a fixed basis $U_b$, we perform the $i$-th Hadamard test with $M_H m_i / M$ measurements, where $m_i$ is the times of appearances of the $\ket i$ in the sampling step. 
With this budget, we obtain the following error bounds.
\begin{lemma}[Measurement errors]\label{le.hadamard_errors}
Let $\ket\psi$ be a quantum state, and let $\bm a$ be coefficients defined in~\Cref{eq.measurable_element}. 
The estimation of one element is conducted via Hadamard tests of $M_H$ samples, and both real and imaginary parts are obtained. 
For the error between the estimate $\bm{\hat a}$ and the true value $\bm a$ it holds that
\begin{equation}
\Vert {\bm a} - \hat{\bm a}\Vert_2^2 \leq 2\sqrt 2 K B \sqrt\frac{\log 2/\delta}{M_H}.
\end{equation}
\end{lemma}
The proof can be found in~\Cref{app.hadamard_errors}.

\paragraph{Computing overlaps}{
The next step is to estimate the values $F_{i_b, j_a}$ to obtain the coefficients in the MBR. 
Notice that only the relevant overlaps are needed, thus incurring in a computational cost of $\mathcal{O}( (K\ B) \times C)$, where $C$ is the cost associated to computing one overlap. 
Assuming we have access to a quantum computer capable of implementing the changes of basis given by $\{U_b\}$, computing the real and imaginary parts of $F_{i_b, j_a}$ is efficient via Hadamard tests in the case $F_{i_a, j_b} \in \Omega(1/\Opoly n )$, since the measurement budget scales as $\vert F_{i_b, j_a}\vert^2$ for constant multiplicative error. In fact, analogous to~\Cref{le.hadamard_errors}, the error for the overlap matrix will scale as
\begin{equation}
    \Vert{F - \hat F}\Vert_2^2 \leq 2\sqrt{2} K^2 B^2 \sqrt{\frac{\log{2/\delta}}{M_H}}
\end{equation}
In certain cases, it is sometimes possible to compute the overlaps via classical algorithms, even when these overlaps decay exponentially with the number of qubits. 
For more information, we refer to~\Cref{sec.mbr}.}

\paragraph{Overlap inversion}
The only calculation left is computing and inverting $F$ to obtain $\bm\alpha$. 
Notice that we can easily extract the values of $\alpha_b$ by imposing the normalization constraint in~\Cref{eq.normalization}. 
The inversion of $F$ is in principle only possible under the assumption that no vector $U_b\ket{i_b}$ is linearly dependent with the rest of the considered states. 
In the case this constraint is not fulfilled, it means that some measurements are redundant, and can therefore be neglected, and no fundamental caveat arises. 
We can easily circumvent this problem through the pseudo-inverse of $F$~\cite{penrose1955generalized}. 
This step outputs the coefficients
\begin{equation}
\bm \alpha = F^{-1} \bm a, 
\end{equation}
which translates into the sets of coefficients $\left(\{\alpha_b\}, \{c^{(b)}_{i_b}\}\right)$. Standard error propagation bounds the errors in $\bm\alpha$ to first order as
\begin{multline}
    \Vert{{\bm\alpha} - \hat{\bm\alpha}}\Vert_2 \leq \\ \Vert{F^{-1}} \Vert_2 \Vert {\bm a} - \hat{\bm a}\Vert_2 + 
    \Vert{F^{-1}} \Vert^2_2 \Vert{F - \hat F}\Vert_2 \Vert\hat{\bm a}\Vert_2 .
\end{multline}
These errors strongly depend on the condition of the overlap matrix $F$. Notice that $F$ is a Gramian matrix, therefore it will become ill-conditioned if the considered bases are linearly dependent with each other. In particular, if $\mathcal F_{K, U_b^\dagger U_a} \sim 0$ for all pairs $(a, b), a\neq b$, the error scale as the measurement error. Finally, the values $\bm\alpha$ are fully specified if the normalization condition from~\Cref{eq.normalization} is imposed. 

\section{Conclusions}\label{sec.conclusions}

\para{In this work we explore the MBR}{
This work introduces the multiple-basis representation (MBR) of quantum states. 
The novelty of MBR consists in using several bases simultaneously, chosen in such a way that the state of interest has a significant overlap with some state with small \cbe-rank, when expressed in the corresponding basis. 
The choice of bases determines whether the MBR can be used to compute overlaps using only classical resources, or quantum computers are required. 
We further provide a classification of subclasses of MBR depending on the properties of these unitaries and distill conditions under which quantum computers are needed to perform classical computations of expectation values. 
We also discuss the optimality of the choice of bases to apply MBR, and prove that mutually unbiased bases (MUB) minimize the relative overlaps among the components associated to each basis. We also characterize several properties of MBR under the MUB assumption. 
}

This manuscript also gives several possible applications for the MBR. 
The first application is using MBR to find approximations of ground states of physical systems. 
The results here depicted show comparable performance to tensor-network methods. 
We foresee that our method can capture long-range correlations more accurately than standard algorithms. 
The second application is the merge of MBR with the \ryc\ method to simulate deep circuits through classical combination of shallow circuits. 
We provide rigorous upper and lower bounds to the depth savings obtained by merging these two techniques together. 
Finally, we outline a tomographical protocol based on the MBR representation for pure quantum states, when many copies of this state can be processed through a quantum computer, providing rigurous bounds on the attained error.

We envision this work to open the path to systematically use several bases simultaneously to represent quantum states, and exploit their advantages. In particular, we envision these bases to be implementable with limited-depth quantum circuits, thus being executable in near-term quantum computers. 
Even with constant-depth quantum circuits, the combinatorial magnitudes arising when considering all possible MBRs, the question becomes how to find useful bases, rather than searching for optimal ones. Hence, any approach to this problem is inherently heuristic. 
As shown in the ground state approximation problem addressed in this work, physically motivated choices lead to classically tractable ways to select bases, and opens the path towards efficiently solving problems using heuristic methods.

The ultimate goal of MBR is to help understanding, and ideally pushing, the boundaries of what can be efficiently done with a classical computer, or with quantum-assisted classical approaches. 
To this end, this work provides conditions under which computing expectation values of MBR-states is classically simulable. 
The MBR efficiently represents quantum states with certain properties, along the lines of tensor networks for states with limited entanglement or stabilizer states for circuits with Clifford gates. 
Since the MBR can be utilized with and without assistance of a quantum computer, it allows for studying hybrid representations of quantum states. This is left as an open research to be addressed with access to experimental resources.

\section*{Data Availability}
The code to simulate the numerical data in this manuscript is available online \cite{github}, as well as the data shown in plots \cite{data}.

\acknowledgements
A. P.-S acknowledges fruitful discussions with Daniel Pérez-Salinas. 
The authors achknowledge comments in the manuscript from Xavier Bonet-Monroig. 
The authors would like to extend their gratitude to all members of aQa Leiden for fruitful discussions.
P.E. acknowledges the support received through the NWO-Quantum Technology
programme (Grant No.~NGF.1623.23.006).
J.T. acknowledges the support received from the European Union's Horizon Europe research and innovation programme through the ERC StG FINE-TEA-SQUAD (Grant No.~101040729). 
This work was supported by the Dutch National Growth Fund (NGF), as part of the Quantum Delta NL programme, and also funded by the European Union under Grant Agreement 101080142 and the project EQUALITY.
This work was also partially supported by the Dutch Research Council (NWO/OCW), as part of the Quantum Software Consortium programme (project number 024.003.03), and co-funded by the European Union (ERC CoG, BeMAIQuantum, 101124342).
This publication is part of the ``Quantum Inspire - the Dutch Quantum Computer in the Cloud" project (with Project No. NWA.1292.19.194) of the NWA research program ``Research on Routes by Consortia (ORC)", which is funded by the Netherlands Organization for Scientific Research (NWO).

The views and opinions expressed here are solely
those of the authors and do not necessarily reflect those of the funding institutions. Neither
of the funding institutions can be held responsible for them.

\bibliography{references.bib}

\onecolumngrid
\appendix

\subsection{Proof of \texorpdfstring{\Cref{le.exp-value}}{Lemma 1}}\label{app.exp-value}
\begin{proof}
   The expectaction of a MBR state with respect to an observable is simply
\begin{equation}\label{eq.exp_value}
\bra{\psi_{K, B}} \hat O \ket{\psi_{K, B}} = 
\sum_{b = 1}^B  \sum_{a = 1}^B 
\sum_{\substack{i_b \in S_b \\ \vert S_b \vert = K}} 
\sum_{\substack{j_a \in S_a \\ \vert S_a \vert = K}} 
\alpha_a \alpha_b \left(c^{(b)}_{i_b}\right)^* c^{(a)}_{j_a} 
\bra{i_b} U^\dagger_b \hat O U_a \ket{j_a}.
\end{equation}
The cost of this calculation scales as $\Opoly{K, B} \times C = \Opoly n \times C$, where $C$ is the cost of computing one element, that is 
\begin{equation}\label{eq.single_element}
\hat O_{i_b, j_a} = \bra{i_b} U^\dagger_b \hat O U_a \ket{j_a}.
\end{equation} 
If this operation can be handled with classical resources, that is $C \in \Opoly n$, then computing expectation values is doable with classical resources.  
\end{proof}

\section{Proof of \texorpdfstring{\Cref{th.venn}}{Theorem 1}}\label{app.relations}
We provide an extended version of the theorem here. 
\begin{theorem}
    The relationship among the MBR subclasses and other representations are the following:
    \begin{align}
    \stab & \subsetneq \mbrc \\ 
    \mps & \subsetneq \mbrc \\ 
    \stab & \not\subset \mbrclog \\ 
    \stab & \cap \mbrclog \neq \emptyset \\
    \mps & \cap \mbrclog \neq \emptyset \\ 
    \mps & \not\subset \stab \\
    \mps & \cap \stab \neq \emptyset \\ 
    \mbrmub & \cap \stab \neq \emptyset \\
    \mbrmub & \cap \mps \neq \emptyset \\
    \mbrmub & \cap (\mbrc \cap \stab) \neq \emptyset \\
    \mbrmub & \cap (\mbrq\cap\mbrc) = \emptyset 
\end{align}
\end{theorem}
\begin{proof}
  We proof each statement separately. 

\paragraph*{\texorpdfstring{$\stab \subsetneq \mbrc$:}{STAB in MBR_C}} The \stab\ formalism deals with Clifford circuits, possibly doped with at most $\mathcal O(\log n)$ non-Clifford gates. 
Computing $\bra{i} U^\dagger \hat O V \ket j$ is possible (assuming $U, V$ satisfy the Clifford condition) since $\hat O$ is a combination of $\Opoly n$-many Pauli strings. 
To do so, we just need to realize that $U^\dagger \hat O V$ is a combination of $\Opoly n$-many Clifford circuits, and propagate the stabilizers of the $\ket 0$ state~\cite{gottesman1998heisenberg}, thus showing $\stab \subset \mbrc$. 

To show $\stab \neq \mbrc$ we just need to find an example of circuits for which computing $\bra{i} U^\dagger \hat O V \ket j$ is doable with classical methods outside the stabilizer formalism. 
Choose $U, V$ to be a single-layer of arbitrary unitary rotations, thus $\mathcal O(n)$ non-Clifford gates and outside the reach of the stabilizer formalism. 
However, for $\hat O$ Pauli strings, computing the quantity of interest is possible. 
It is only required to perform $\mathcal O (n)$ single-qubit computations, with linear cost. 

\paragraph{\texorpdfstring{$\mps \subsetneq \mbrc$ }{MPS not in MBRC}} The \mps\ formalism deals with circuits with low entanglement. 
Computing $\bra{i} U^\dagger \hat O V \ket j$ is possible (assuming $U, V$ satisfy the MPS condition) since $\hat O$ is a combination of $\Opoly n$-many Pauli strings. 
To do so, we just need to realize that $U^\dagger \hat O V$ is matrix product operator with bond dimension in $\Opoly n$, thus showing $\mps \subset \mbrc$. 

To show $\mps \neq \mbrc$ we just need to find an example of circuits for which computing $\bra{i} U^\dagger \hat O V \ket j$ is doable with classical methods outside the MPS formalism. 
Assume $\hat O$ to be a Pauli string. Choose $U, V$ to be a Clifford circuit with long-term connectivity. The bond dimension to compute observables using the MPS formalism increases exponentially with the range. However, the stabilizer formalism trivially allows to compute observables for these circuits. . 

\paragraph*{\texorpdfstring{$\stab \not\subset \mbrclog$:}{STAB not in MBRc-log}} Following the reasoning of the previous paragraph, the \stab\ formalism allows us to compute $\bra{i} U^\dagger \hat O V \ket j$ for arbitrary Clifford circuits. 
This includes circuits with $\Opoly n$-depth, not only $\mathcal O(\log n)$. 
Existing results show that it is possible to efficiently reduce any Clifford circuit to another one with at most $\Opoly n$-depth, thus no possible reduction can be applied to circumvent this result. 

\paragraph*{\texorpdfstring{$\stab \cap \mbrclog \neq \emptyset$:}{STAB intersection MBRc-log not empty}} To show this it is enough to find an example of a $\stab$ circuit with logarithmic depth. 
It is enough with restricting the depth of the Clifford circuits. 

\paragraph*{\texorpdfstring{$ \mps \cap \mbrclog \neq \emptyset $:}{MPS intersection MBRc-log not empty}}
\mps\ representation can compute $\bra{i} U^\dagger \hat O V \ket j$ for certain circuits with $\mathcal O(\log n)$-depth. 
For circuits with nearest-neighbours connectivity, each two-qubit gate can be decomposed as a superposition $\mathcal O(1)$ tensor-product operations. 
Recurrent application of this principle allows to find an \mps\ representation with $\chi \in \Opoly n$, thus efficient. 
This proof has been widely explored in the literature~\cite{zhou2020what}. 

It is possible to find examples of \mbrclog\ that are not tractable through \mps. 
As a preliminary result, a state with \cbe\ rank $K$ is \mps izable with bond dimension at most $K$. 
The reason is that if we perform the Schmidt decomposition between any two partitions $(A, B)$, such that $2^{\min(A, B)} > K$, we separate the digits in the binary representation of the integers. 
This provides a Schmidt rank at most $K$, and holds for all partitions. 
Now, if $\{U_b\}$ admits a representation as matrix product operator with bond dimension $\chi$, then $\ket{\psi_b}$ admits an \mps\ representation of bond dimension $\chi(\ket{\psi_b})$, since $\chi(U \ket\psi) \leq \chi(U) \chi(\ket\psi)$. 
In addition, the property $\chi(\ket\psi + \ket\phi) \leq \chi(\ket\phi) + \chi(\ket\phi)$, and therefore the maximal bond dimension is upper bounded, $\chi \in \Opoly n$. 

Additionally, one can find circuits with depth $\omega (\log n)$ that are classically tractable with MPS. It suffices with padding a MPS-representable states with gates that merge into identity. Those fall into the category of $\mbrc$, but not $\mbrclog$. Finally $\mps \in \mbrc$ by definition.  

\paragraph*{\texorpdfstring{$ \mps \not\subset \stab $:}{MPS not in STAB}}
A Clifford circuit with long-range correlations, as $\mathcal O(n)$, even with $\mathcal O(1)$-depth, demands an exponential overhead in the bond dimension $\chi$ of the \mps\ representation, thus making it inefficient.

\paragraph*{\texorpdfstring{$ \mps \cap \stab \neq \emptyset $:}{MPS intersection STAB not empty}}
In alignment with the previous paragraph, any $\mathcal O(\log n)$-circuit with nearest neighbour connectivity admits \mps\ representation, in particular a circuit composed only by Clifford gates.  

\paragraph*{\texorpdfstring{$ \mbrmub \cap \stab \neq \emptyset $:}{MBR-MUB intersection STAB not empty}}
It is possible to construct a complete set of MUB using only Clifford circuits~\cite{kern2010complete}. 

\paragraph*{\texorpdfstring{$\mbrmub \cap \mps \neq \emptyset$:}{MBR-MUB intersection MPS not empty}}
Construct a pair of bases that are locally-MUBs, with locality $\mathcal O(1)$. These are by construction MUBs, and quantum circuits of depth $\mathcal O(1)$, which are representable by MPS. 

\paragraph*{\texorpdfstring{$ \mbrmub \cap (\mbrc \cap \stab) \neq \emptyset $:}{MBR-MUB intersection (MBR-C intersection STAB) not emtpy}}
It is not needed to have Clifford circuits to obtain MUBs. 
One can consider, for instance, two unitaries given by single-qubit arbitrary rotations with the constraint that they form |MUBs. 
Since they are single-qubit unitaries, computing $\bra{i} R^\dagger U^\dagger \hat O V R \ket j$ becomes trivial, with $\mathcal{O}(n)$ non-Clifford gates. 

\paragraph*{\texorpdfstring{$ \mbrmub \cap (\mbrq\cap\mbrc) = \emptyset $:}{MBR intersection (MBR-Q intersection MBR-C) is empty}}
Assume that there are sets of MUBs that are not computable classically. 
Then, one needs to resort to quantum computers to compute $\ket{i} U^\dagger \hat O V \ket j$, in particular for $\hat O = I$. 
In that case, this quantity is by definition (in amplitude) $N^{-1} = 2^{-n}$. 
Quantum methods would require an exponential measurement budget to resolve the relative phases. 

\end{proof}

\section{Proof of \texorpdfstring{\Cref{le.mubs}}{Lemma 2}}\label{app.mubs}
The quantity of interest can be computed as
\begin{equation}
\mathcal F_K = \min_{U \in \mathcal{SU}(2^n)}  \max_{\substack{\ket{\psi_1}, \ket{\psi_2} \\\cbe\left(\ket{\psi_{\{1, 2 \}}}\right) = K}} \vert\bra{\psi_1}U\ket{\psi_2} \vert. 
\end{equation}
More explicitly, we can rewrite it as
\begin{equation}
\mathcal F_K = \min_{U \in \mathcal{SU}(2^n)} \max_{\substack{ \{ c^{(b)}_i, b \in \{1, 2\}\}\\ S_1, S_2; \vert S_b \vert = K}} \left\vert d^*_j c_i \bra{j} U \ket i \right\vert.
\end{equation}
Since we can optimize over $S_b$, given a fixed $U$ and set $\{ c_i, d_i \}$ it is always possible to choose the sparse sets such that the largest coefficients correspond to the largest overlaps. 
Thus, the operation $U$ must fulfill a condition that $\mathcal F_{U, K}$ is independent of the sets $S_{b}$. 
Such requirement is met when $U$ defines a change of bases between MUBs. 
Finally, for any pair of sets $S_1, S_2$, it is always possible to choose the coefficients $\{ c_i, d_i \}$ such that they perfectly align with the overlap between bases given by $\bra{j} U \ket i$. 
Notice that the phases can be complex. 
Thus, the maximization is solved by choosing equally distributed coefficients $\{ c_i, d_i \}$ aligned with the maximum values of $\bra{j} U \ket i$, and the minimization problem is solved by making all overlaps equal
\begin{equation}
\vert \bra i U \ket j\vert = \vert \bra{i^\prime} U \ket{j^\prime}\vert.
\end{equation}
Normalization conditions imposed by unitarity of $U$ suffice to obtain
\begin{equation}
\vert \bra i U \ket j\vert = \frac{1}{\sqrt{2^n}}
\end{equation}

The relative overlap, with respect to the change of basis $U$, is given by
\begin{equation}
\vert \bra{{\psi_2}} U \ket{\psi_1} \vert = \left\vert\sum_{j\in S_2}  \sum_{i \in S_1} d^*_j c_i \bra{j} U \ket i \right\vert.
\end{equation}
We recall $\vert S_{b}\vert = K$. 
Using the triangular inequality we can bound
\begin{equation}
\vert \bra{{\psi_2}} U \ket{\psi_1} \vert \leq \\ 
\sum_{j\in S_2}  \sum_{i \in S_1} \vert d_j \vert \vert c_i \vert \left\vert\bra{j} U \ket i \right\vert.
\end{equation}
Each piece can be bounded separatedly. 
First, by assumption $U$ defines a change between MUB, and thus the overlap between all elements is given by $2^{-n/2}$. 
Secondly, the coefficients $\{c_i\}$ are subject to the condition of $\sum_i \vert c_i\vert^2 = 1$, and equivalently for $d$. 
Therefore, the maximum possible overlap is attained when the quantum state is equally distributed among all the considered elements. 
Thus,
\begin{equation}
\vert \bra{{\psi_2}} U \ket{\psi_1} \vert \leq \sum_{j\in S_2}  \sum_{i \in S_1} \frac{1}{K \sqrt{2^n}} = \frac{K}{\sqrt{2^n}},
\end{equation}
leading to the desired result. \qed

\section{Proof of \texorpdfstring{\Cref{le.bounds_error}}{Lemma 3}}\label{app.bounds_error}

We start by writing the MBR representation. 
The index $K$ is ommited for convenience. 
\begin{equation}
\ket{\psi_B} = \sum_b \alpha_b \ket{\psi_b}.  
\end{equation}
The overlap of $\ket{\psi_B}$ with the target state $\ket\psi$ is given by
\begin{equation}\label{cor.overlap_app}
\left\vert \braket{\psi_B}{\psi}\right\vert = \left\vert \sum_b \alpha_b \braket{\psi}{\psi_b}\right\vert.
\end{equation}
This overlap is maximized by choosing $\alpha_b$ in such a way that it aligns with the partial overlaps $\braket{\psi}{\psi_b}$. 
Since we are allowed to choose $\alpha$ in any way, under this condition
\begin{equation}
\left\vert \braket{\psi_B}{\psi}\right\vert = \frac{1}{N}\sum_b \left\vert \braket{\psi}{\psi_b}\right\vert^2,
\end{equation}
with $N$ being a normalization constant. 
We can fix its value by imposing the condition
\begin{equation}\label{eq.normalization_app}
1 = \left\vert \braket{\psi_B}{\psi_B}\right\vert = \frac{1}{N^2}\sum_{a, b} \braket{\psi_b}{\psi}\braket{\psi}{\psi_a}\braket{\psi_a}{\psi_b}.
\end{equation}
Separating terms with $a = b$, we obtain
\begin{equation}
\left\vert \braket{\psi_B}{\psi_B}\right\vert = \\ \frac{1}{N^2}\left(\sum_{b} \vert\braket{\psi_b}{\psi}\vert^2 + \sum_{a \neq b}  \braket{\psi_b}{\psi}\braket{\psi}{\psi_a}\braket{\psi_a}{\psi_b}\right).
\end{equation}
Using~\Cref{le.mubs}, we can bound the second term as
\begin{equation}
\left\vert \sum_{a \neq b}  \braket{\psi_b}{\psi}\braket{\psi}{\psi_a}\braket{\psi_a}{\psi_b}\right\vert \leq \\ \frac{K}{\sqrt{2^n}} \left\vert\sum_{a \neq b}  \braket{\psi_b}{\psi}\braket{\psi}{\psi_a}\right\vert.
\end{equation}
The elements inside the sum can be rewritten as
\begin{equation}
\left\vert \sum_{a \neq b}  \braket{\psi_b}{\psi}\braket{\psi}{\psi_a}\right\vert = \\ \left\vert \sum_{a,b}  \braket{\psi_b}{\psi}\braket{\psi}{\psi_a} - \sum_{b}  \braket{\psi_b}{\psi}\braket{\psi}{\psi_b}\right\vert,
\end{equation}
and each of these terms is respectively the squared 1-norm and 2-norm of the vector specified by the overlaps. 
Since the $\Vert \cdot \Vert_1 \leq \sqrt B \Vert \cdot \Vert_2$, we can further simplify the previous equations to
\begin{equation}
\left\vert \sum_{a \neq b}  \braket{\psi_b}{\psi}\braket{\psi}{\psi_a}\braket{\psi_a}{\psi_b}\right\vert \leq \frac{K B}{\sqrt{2^n}} \sum_{b} \left\vert \braket{\psi_b}{\psi} \right\vert^2.
\end{equation}
Plugging this result into~\Cref{eq.normalization_app} we obtain
\begin{equation}
N^2 = \sqrt{\sum_{b} \left\vert \braket{\psi_b}{\psi} \right\vert^2}\left(1 + z \frac{K B}{\sqrt{2^n}} \right),
\end{equation}
for some $z\in [-1, 1]$. 
Using this result in~\Cref{cor.overlap_app} and using the Taylor approximation in the $KB \sqrt{2^{-n}}$ term, we obtain the desired result. 
\qed

\section{Proof of \texorpdfstring{\Cref{le.dimensionality}}{Lemma 4}}\label{app.dimensionality}
Consider $B$ different $K$-sparse states $\ket{\psi_b}$, and $B$ changes of bases $\{U_b\}_{b \leq B}$ each defining a basis used to to express the corresponding state. 
We are interested in finding the dimension of the space opened by $U_b \ket{\psi_b}$.

We can argue by Gram-Schmidt orthogonalization process~\cite{cheney2009linear}. 
The input of the usual formulation of the Gram-Schmidt process is a set of non-orthogonal vectors, and the output is another set of orthonormal vectors spanning the same linear space. 
The process is deterministic, provided the input vectors are ordered. 
The key step of Gram-Schmidt is subtracting from one vector its own projection over the previously considered one, so that each vector is only kept if it opens a new dimension. 

In our case, the input set is substituted by the spaces opened by sparse states with different bases. 
Let $\mathcal U_b$ be the linear spaced spanned by $\{ U_b \ket i\}_{i \in S_b}$, where $\vert S_b \vert = K$. 
The corresponding projector is $\Pi_b = U \left( \sum_{i \in S_b} \ket i \bra i \right) U^\dagger$. 
We apply Gram-Schmidt orthogonalization to the set $\{ \mathcal U_b\}_b$ to obtain a set of orthogonal spaces $\{ \mathcal V_b\}_b$. 
We do not consider normalization in this step. 
Each new orthogonal space is given by
\begin{equation}
\mathcal V_b = \mathcal U_b - \sum_{a = 1}^{b - 1} \Tr{\mathcal U_b \Pi_a}	 \mathcal U_a.
\end{equation}
By recalling~\Cref{le.mubs}, we know that the overlap between any two $K$-sparse states each one in a MUB is bounded by $K 2^{-n/2}$. 
Therefore, it is guaranteed that $\mathcal V_b$ has a non-zero dimension as long as 
\begin{equation}
\frac{K (B-1)}{\sqrt{2^n}} \leq 1,
\end{equation}
leading to the desired result. \qed

\section{Proof of \texorpdfstring{\Cref{le.volume_sparse}}{Lemma 5}}\label{app.volume_sparse}
We are interested in computing the fraction of states that is $\epsilon$-close to a $K$ sparse state. 
To do so, we need to integrate the surface of the sphere for quantum states within the desired distance. 
To do so we just need to recall that a the amplitudes of a $n$-qubit state can be understood as points on a sphere of dimension $2^{n + 1}$. 
Correspondingly, all states with $CB$-rank $K$ can be understood as the unit sphere of dimensions $2K$. 
The fraction of states within $\epsilon$ distance to any state of $CB$-rank $K$ is the surface of the hypersphere around the $(2K - 1)$-dimensional equator. 
The assumed Haar randomness of Haar allows us to treat the corresponding subspaces with the beta distribution~\cite{bailey1992distributional}. 
Let $(x_1, x_2, \ldots, x_N)$ be the coordinates of points on a unit sphere, thus $\sum_i x_i^2 = 1$. 
We sum a subset of these variables as $X_K = \sum_{i = 1}^{2K} x_i^2$. 
Then $X_K$ follows a random distribution
\begin{equation}
X_K \sim \operatorname{Beta}\left(K, N - K\right). 
\end{equation}
The corresponding surface is obtained by computing the probabilities of $X_K \geq 1 - \epsilon$, thus
\begin{equation}
S_K(\epsilon) = \int_{0}^\epsilon dt\ t^{2^n - K - 1} (1 - t)^{K - 1}.
\end{equation}
The normalization is computed by dividing over the total surface, for $\epsilon = 1$. This normalization is known as the regularized incomplete beta function, detailed in the next section. \qed

\section{Regularized incomplete beta function}\label{app.beta}

The beta function is defined as
\begin{equation}
    B(\alpha, \beta) = \int_{0}^1 t^{1 - \alpha} (1-t)^{1 - \beta} dt = \frac{\Gamma(\alpha)\Gamma(\beta)}{\Gamma(\alpha + \beta)},
\end{equation}
where $\Gamma(\cdot)$ is the Euler gamma function, functioning as an extension to the body of complex numbers of the factorial, that is
\begin{align}
    \Gamma(z) & = \int_{0}^\infty t^{z - 1} e^{-t} dt \\
    \Gamma(n) & = (n - 1)! ; \quad n \in \mathbb{N}.
\end{align}
The incomplete beta function is defined as
\begin{equation}
    B(x; \alpha, \beta) = \int_{0}^x t^{1 - \alpha} (1-t)^{1 - \beta} dt .
\end{equation}
It is commonly regularized to adapt this function to the beta distribution~\cite{bailey1992distributional} as
\begin{equation}
    I(x; \alpha, \beta) = \frac{\Gamma(\alpha + \beta)}{\Gamma(\alpha)\Gamma(\beta)}\int_{0}^x t^{1 - \alpha} (1-t)^{1 - \beta} dt .
\end{equation}

\section{Proof of \texorpdfstring{\Cref{le.lcu}}{Lemma 7}}\label{app.lcu}
We begin by detailing the linear combination of unitaries. 

We first apply an auxiliary operation $W$ on an ancilla register of size $\log B$ with the effect
\begin{equation}
W\ket 0 = \frac{1}{\sqrt{\Vert \alpha \Vert_1}}\sum_b \sqrt{\alpha_b} \ket b. 
\end{equation}
This register controls the application of unitary operations generating each state $U_b \ket{\psi_b}$. 
By reverting the $W$ operator we end up with a state
\begin{equation}
\ket{\Psi} = \frac{1}{\Vert \alpha \Vert_1} \ket 0 \ket{\psi_{K, B}} + \sqrt{1 - \frac{\Vert \ket{\psi_{K, B}}\Vert^2}{\Vert \alpha \Vert_1^2}} \ket 1 \ket \eta,
\end{equation}
with $\ket{\psi_{K, B}} = \sum_b \alpha_b U_b \ket{\psi_b}$ as in~\Cref{eq.state_def}. 
By creating $\ket\Psi$ and post-selecting on $\ket 0$ in the first register, we obtain $\ket\psi$ in the second register with a certain success probability $P_{\rm success}$. 

We are only interested in the first term of the equation above, marked with $\ket 0$, since it has the state we are interested in creating. 
The overlap of $\ket\Psi$ with $\ket 0 \ket\psi$ is given by
\begin{equation}
\vert \braket{\Psi}{\psi_{K, B}}\vert^2 = \frac{\Vert \ket{\psi_{K, B}} \Vert^2}{\Vert \alpha \Vert_1^2}.
\end{equation}
Such success probability can be readily upper bounded as
\begin{equation}
\Vert \ket\psi \Vert^2 = \sum_b \alpha_b^2 \braket{\psi_b}{\psi_b} + \sum_{a\neq b} \alpha_a \alpha_b \bra{\psi_a}U_a^\dagger U_b \ket{\psi_b}.  
\end{equation}
We consider now the block $F_{a, b}$ of the overlap matrix. 
The corresponding element can be upper bounded as
\begin{equation}
\bra{\psi_a}U_a^\dagger U_b \ket{\psi_b} \leq \Vert F_{a, b}\Vert_\infty.
\end{equation}
On the other hand, we recall 
\begin{equation}
\sum_{a \neq b}\alpha_a \alpha_b = \sum_{a, b} \alpha_a \alpha_b - \sum_b \alpha_b^2 = \Vert \alpha \Vert_1^2 - \Vert \alpha \Vert_2^2.
\end{equation}
Therefore, we can further simplify
\begin{equation}
\vert \braket{\Psi}{\psi_{K, B}}\vert^2 \leq \frac{\Vert \bm \alpha \Vert_2^2}{\Vert \bm \alpha \Vert_1^2}\left( 1 - \maxF \right) + \maxF.
\end{equation}

\qed

\section{Proof of \texorpdfstring{\Cref{le.upperboundsmbr}}{Lemma 8}}\label{app.upperboundsmbr}
The starting point is the fixed-point amplification procedure from Ref.~\cite{yoder2014fixedpoint}. 
It is possible to load a state $\delta$-close to $\ket\psi$ with a depth lower-bounded by
\begin{equation}
L\geq \frac{\log(2 \delta^{-1})}{\sqrt{P_{\rm success}}}.
\end{equation}
Using the results from~\Cref{le.lcu} it is straightforward to see that
\begin{equation}
L \geq \log(2 \delta^{-1}) \frac{\Vert \bm\alpha\Vert_1}{\sqrt B\Vert \bm\alpha\Vert_\infty} \left(1 +  \maxF \frac{\Vert \bm\alpha\Vert_1^2}{\Vert \bm\alpha\Vert_2^2}\right)^{-1/2}.
\end{equation}
Making use of Taylor's expansion and the property
\begin{equation}
\frac{1}{\sqrt{1-x}} > 1 - \frac{x}{2}\quad \forall x > 0, 
\end{equation}
we can bound
\begin{equation}
L \geq \log(2 \delta^{-1}) \frac{\Vert \bm\alpha\Vert_1}{\sqrt B\Vert \bm\alpha\Vert_\infty}\left( 1 -  \frac{\maxF}{2} \frac{\Vert \bm\alpha\Vert_1^2}{\Vert \bm\alpha\Vert_2^2}\right).
\end{equation}

Fixed-point amplification is protected against overcooking. 
Thus, we can apply a circuit with depth twice the previous lower-bound and maintain the guarantees. 

From this point on, we just need to count. 
Each sparse state can be created in depth at most $d(K)$, and each basis change can be applied in at most $d(U)$. 
This is repeated for each basis, leading to a depth $(d(U) + d(K) ) B$. 
The fixed point amplification is given by $L$ repetitions of the same circuits. 
Additionally, we need to count the depth of the operation $W$, $d(W)$. 
In total, we obtain a depth upper bounded by
\begin{equation}
D \leq d(W) + 2 \ (d(U) + d(K)) \sqrt{B} \log(2/\delta)\frac{\Vert \bm\alpha\Vert_1^2}{\Vert \bm\alpha\Vert_2^2}\\  \left(1 - \frac{F}{2}\left( \frac{\Vert \bm\alpha\Vert_1^2}{\Vert \bm\alpha\Vert_2^2} + 1\right) \right).
\end{equation}
\qed

\section{Proof of \texorpdfstring{\Cref{th.concentration}}{Theorem 2}}\label{app.concentration}
We start by considering a probability distribution with $N$ possible different outcomes. 
We sample this probability distribution $M$ times. 
Let $p_i$ be the probability of sampling outcome $i$, and let $\hat p_i = m_i / M$ be an estimator given by sampling. 
The number $m_i$ is the number of appearances of the outcomes $i$. 
Without loss of generality, we assume $\hat p_{i} \geq \hat p_{i+1}$. 
Notice we cannot in principle assume any ordering in the true values $p_{i}$ since those are not exactly accessible. 
For an integer $K$ we define
\begin{equation}
\hat S_K = \sum_{i = 1}^{K} \hat p_i, 
\end{equation}
as the estimator of the joint probability of the $K$ largest estimators of the probability. 

The effect of these estimators $\hat S_K$ is to reduce a multinomial distribution into a binomial one, in which the outcomes are split in two classes. 
After splitting, we can estimate $\hat S_K$ using Wilson estimator~\cite{wilson1927probable}. 
The reason for Wilson estimator, instead of some other more standard versions such as Wald's~\cite{laplace1820theorie}, is that Wilson produces reliable confidence intervals even in the case of small $\hat p_i$ or $\hat S_K$. 
According to Wilson interval 
\begin{equation}
S_K = \frac{1}{1 + \frac{z^2}{M}}\left(\hat S_K + \frac{z^2}{2M} \right) \pm \frac{z}{1 + \frac{z^2}{M}} \sqrt{\frac{\hat S_K (1 - \hat S_K)}{M} + \frac{z^2}{4M}}, 
\end{equation} 
with confidence probability $1 - \delta$, and $z$ is the $\delta / 2$-quantile of the standard normal distribution. 

We consider now the estimators $\hat S_K$, $\hat S_{K + 1}$. 
Each estimator gives rise to a probability distributions of compatible values of $S_K / S_{K + 1}$. 
We need to obtain the requirements in $\hat p_{K+1}$ to avoid overlapping distributions of $S_K / S_{K + 1}$. 
In the extremal case, since $\hat S_K < \hat S_{K +1}$, we have with probability at least $1 - \delta / 2$
\begin{align}
S_K \leq & \qquad \frac{1}{1 + \frac{z^2}{M}}\left(\hat S_K + \frac{z^2}{2M} \right) +  \frac{z}{1 + \frac{z^2}{M}} \sqrt{\frac{\hat S_K (1 - \hat S_K)}{M} + \frac{z^2}{4M}}, \\
S_{K+1} \leq & \qquad \frac{1}{1 + \frac{z^2}{M}}\left(\hat S_{K+1} + \frac{z^2}{2M} \right) -  \frac{z}{1 + \frac{z^2}{M}} \sqrt{\frac{\hat S_{K+1} (1 - \hat S_{K+1})}{M} + \frac{z^2}{4M}}. 
\end{align}

Imposing $S_K < S_{K +1}$, 
\begin{equation}
\hat p_{K+1} > \\ z \left( \sqrt{\frac{\hat S_{K+1} (1 - \hat S_{K+1})}{M} + \frac{z^2}{4M}} + \right.\\ \left. \sqrt{\frac{\hat S_{K} (1 - \hat S_{K})}{M} + \frac{z^2}{4M}} \right).
\end{equation}
We can further simplify the previous equation with $0 \leq \hat S_{K+1} \leq 1$ to obtain
\begin{equation}
\hat p_{K+1} > \frac{1}{\sqrt M} \left(z \sqrt{4 + z^2} \right) > \frac{z^2}{\sqrt M}. 
\end{equation}
The last step is to relate $z$ and $\delta$ in a simpler way. 
First, we recall standard definitions 
\begin{equation}
\frac{\delta}{2} = 1 - \erf(z / \sqrt{2}).
\end{equation}
For sufficiently large $z$, e.g. $z \geq 2$, which corresponds to small $\delta \leq 0.1$, it is possible to bound
\begin{equation}
1 - \erf(z / \sqrt{2}) \geq e^{-z^2}, 
\end{equation}
which readily transforms into
\begin{equation}
z^2 \geq \log(2 / \delta).
\end{equation}

Rearranging the equations, we arrive to the desired result
\begin{equation}
\hat p_{K+1} \geq \frac{\log(2/\delta)}{\sqrt{M}}. 
\end{equation}
\qed

\section{Proof of \texorpdfstring{\Cref{le.hadamard_errors}}{Lemma 9}}\label{app.hadamard_errors}
We begin by computing the distance between the true and estimated values of the $a$ coefficients for one basis, namely
\begin{equation}
\Delta a^{(b)}_{i_b} \equiv \vert a^{(b)}_{i_b} - \hat a^{(b)}_{i_b} \vert^2.
\Vert {\bm a^{(b)}} - \hat{\bm a^{(b)}}\Vert_2^2 \leq \sum_{b \in B} \sum_{i \in S_b} \vert a^{(b)}_{i_b} - \hat a^{(b)}_{i_b} \vert^2.
\end{equation}

Assume that for every $i$ we use a different number of measurements $M_i$, evenly dedicated to real and imaginary parts via Hadamard tests. 
Using Wald estimators~\cite{laplace1820theorie} we obtain with probability $1 - \delta$
\begin{equation}
\Re(a) = \Re(\hat a) \pm z\sqrt{\frac{1 - \Re(\hat\alpha)^2}{2M_i}},
\end{equation}
with $z$ being the $1 - \delta/2$ quantile of the normal distribution. 
This readily implies
\begin{equation}
\Delta \Re(a) \leq \frac{z}{\sqrt{2 M_i}},
\end{equation}
and equivalently for the imaginary part. 

We recall now standard error propagation to bound the joint error of real and imaginary parts as
\begin{equation}
\Delta \vert a \vert^2 = 2 \vert\Re(\hat a)\vert \Delta \Re(\hat a) + 2 \vert\Im(\hat a)\vert \Delta \Im(\hat a) \leq \\
 \left( \vert\Re(\hat a)\vert +\vert\Im(\hat a)\vert \right) \frac{\sqrt 2 z}{\sqrt{M_i}}.
\end{equation}
The Chauchy-Schwarz inequality for $d$-dimensional vectors implies $\norm{a}_1 \leq \sqrt{d} \norm{a}_2$, which immediately implies $\left( \vert\Re(\hat a)\vert +\vert\Im(\hat a)\vert \right) \leq \sqrt{2} \sqrt{\hat p_i}$. 
We can exploit the fact that we have the estimation $\sqrt{\hat p_i} = m_i / \hat S_k$ from previous sampling, thus implying
\begin{equation}
\Delta p \leq \sqrt\frac{\hat p_i}{S_K}\frac{2 z}{\sqrt{M_i}}.
\end{equation}
By selecting $M_i = M \hat p_i / \hat S_K$, we obtain a bound
\begin{equation}
\Delta \vert a \vert^2 \leq \frac{ 2 z}{\sqrt M}. 
\end{equation} 
We can readily bound 
\begin{equation}
\Vert {\bm a^{(b)}} - \hat{\bm a^{(b)}}\Vert_2^2 \leq  \frac{ 2 z K}{\sqrt M}.
\end{equation}

Identifying
\begin{equation}
\delta \leq 2 e^{-z^2 / 2}, 
\end{equation}
we can rewrite
\begin{equation}
\Vert {\bm a^{(b)}} - \hat{\bm a^{(b)}}\Vert_2^2 \leq 2\sqrt 2 K B \sqrt\frac{\log 2/\delta}{M},
\end{equation}
obtaining the desired result. 
\qed
\end{document}